\documentclass[11pt,a4 paper]{article}
\newtheorem{theo}{Theorem}[section]
\newtheorem{cor}{Corollary}[section]
\newtheorem{rem}{Remark}[section]
\newtheorem{defi}{Definition}[section]
\newtheorem{lemma}{Lemma}[section]

\newtheorem{prop}{Proposition}[section]
\newtheorem{ex}{Example}[section]

\newtheorem{op}{Open Problem}

\newcommand{\proof}{\noindent{\em Proof.}\quad}

\usepackage{multirow}
\usepackage{hhline}
\usepackage{array}
\usepackage{makecell}

\usepackage{float}
\usepackage{tocloft}
\usepackage{amsfonts,amsmath,amssymb}
\usepackage{multirow, verbatim}
\usepackage{shadow,enumerate}
\usepackage{makeidx}  
\usepackage{graphicx}
\usepackage{color}

\usepackage{mathrsfs}

\newcommand{\FB}{{\mathbb F}_{2}}
\def\whitebox{{\hbox{\hskip 1pt
        \vrule height 6pt depth 1.5pt
        \lower 1.5pt\vbox to 7.5pt{\hrule width
                  3.2pt\vfill\hrule width 3.2pt}%
        \vrule height 6pt depth 1.5pt
        \hskip 1pt } }}
\def\qed{\ifhmode\allowbreak\else\nobreak\fi\hfill\quad\nobreak\whitebox\medbreak}

\begin{document}
\title{Designing plateaued Boolean functions in spectral domain  \\ and their classification}

\author{
S. Hod\v zi\'c \footnote {University of Primorska, FAMNIT, Koper,
Slovenia, e-mail: samir.hodzic@famnit.upr.si}
\and E.~Pasalic\footnote{
University of Primorska, FAMNIT \& IAM, Koper, Slovenia, e-mail: enes.pasalic6@gmail.com}
\and Y. Wei \footnote{
 Guilin University of Electronic Technology, Guilin,
P.R. China,
 e-mail: walker$_{-}$wyz@guet.edu.}
\and F. Zhang \footnote{School of Computer Science and Technology,
China University of
 Mining and Technology, Xuzhou, Jiangsu 221116,  P.R. China, e-mail:
zhfl203@cumt.edu.cn.}
}

\date{}
\maketitle

\begin{abstract}
The design of plateaued functions over $GF(2)^n$, also known as 3-valued Walsh spectra functions (taking the values from the set $\{0, \pm 2^{\lceil \frac{n+s}{2} \rceil}\}$), has been commonly approached by specifying a suitable algebraic normal form  which then induces this particular Walsh spectral characterization.
In this article, we  consider the reversed design method which specifies these functions in the spectral domain by specifying a suitable  allocation of the nonzero spectral values and their signs. We analyze the properties of {\em trivial and nontrivial} plateaued functions (as affine inequivalent distinct subclasses), which are distinguished by their  Walsh support $S_f$ (the subset of $GF(2)^n$ having the nonzero spectral values) in terms of whether it is an affine subspace or not. The former class  exactly corresponds to partially bent functions and admits linear structures, whereas the latter class may contain functions without linear structures. A simple sufficient condition on $S_f$, which ensures the nonexistence of linear structures, is derived  and some generic design methods of nontrivial plateaued functions without linear structures are given.
The extended affine equivalence of plateaued functions is also addressed using the concept of {\em dual} of plateaued functions.
Furthermore, we solve the problem of specifying disjoint spectra (non)trivial plateaued functions of maximal cardinality whose concatenation can be used to construct bent functions in a generic manner. This approach may lead to new classes of bent functions due to large variety of possibilities to select underlying duals that define these disjoint spectra plateaued functions.
An additional method of specifying affine inequivalent plateaued functions, obtained by applying a nonlinear transform to their input domain, is also given.

\end{abstract}

\section{Introduction}


A class of Boolean functions on $GF(2)^n$ characterized by the property that their Walsh spectra is three-valued (more precisely taking values in $\{0,\pm 2^{\frac{n+s}{2}}\}$ for a positive integer $s<n$) is called  $s$-\emph{plateaued} functions  \cite{Zheng}.
The notion of plateaued functions as defined here does not include bent and linear functions (which only have  two different values in their Walsh spectrum), though sometimes in the literature these families are also included. This class of functions has a wide range of applications such as  cryptography, in the design of sequences for communications, and related combinatorics and designs.
For instance, vectorial plateaued functions were efficiently used to provide large sets of orthogonal binary sequences suitable for CDMA (Code Division Multiple Access) applications which resulted in a significant improvement of the number of users per cell in regular hexagonal networks \cite{WeiGuoCDMA}.
The design methods of plateaued functions have been addressed in several works in the past twenty years \cite{CarletAPN,CarletQ,Feng2,SihemPOP2014,THOMCusick2016,SihemM2017AMC,Fengrong2018IT,Zheng2}, but none of these construction methods is generic and furthermore a simple characterization of these functions in the Walsh domain has not been utilized as a construction rationale.  In \cite{CarletAPN}, Carlet provided a survey on the construction methods of Boolean plateaued functions along with  certain characterizations of (vectorial) plateaued functions which extends the initial attempt towards their characterization in \cite{CarletQ}. Nevertheless, these  characterizations of plateaued functions, given in terms of their second order derivatives, autocorrelation properties or power moments of the Walsh transform in \cite{CarletAPN}, are not efficient as generic construction methods.

In addition,  some  methods  for specifying  plateaued functions in the algebraic normal form domain was considered in \cite{THOMCusick2016} and \cite{Logatchev}.   In \cite{SihemM2017AMC}, a secondary construction of plateaued functions was presented based on the use of three suitable  bent functions.
The importance of plateaued functions also relates to construction of highly nonlinear resilient functions and for this purpose  the so-called disjoint spectra plateaued functions  are required
\cite{Zhang2009}. Certain methods, mainly employing the (generalized) indirect sum method of Carlet \cite{SecEnf},  for generating the sets of disjoint spectra plateaued functions without linear structures were investigated  in \cite{Feng2,Fengrong2018IT}.
%
We briefly mention that the notion of plateaued functions has been generalized over different mathematical structures, the interested reader is referred to   \cite{AycaWilffried,SihemPOP2014,Jong,SihemPlate,SihemM2017IT}.

The main objective of this article is to completely revise the design approach by transferring it to the Walsh domain which gives much more flexibility and generality.
Indeed, any $s$-plateaued function $f$ on $GF(2)^n$ uniquely corresponds to its Walsh support, denoted by $S_f$ (a subset of $GF(2)^n$ for which the Walsh spectral values are nonzero $\pm 2^{\frac{n+s}{2}}$), and a sequence of $(\pm 1)$-signs that precisely specifies these  non-zero Walsh coefficients.
The cardinality  of these positive and negative values in the Walsh support  is also known and it corresponds to the cardinality of zeros and ones of a bent function defined on $GF(2)^{n-s}$. This implies that an alternative approach to design plateaued functions on $GF(2)^n$ is to provide a proper placement of the spectral values  $\{0,\pm 2^{\frac{n+s}{2}}\}$ over the vector space $GF(2)^n$ so that this placement indeed yields a Boolean function (which is then plateaued). A useful categorization of  plateaued functions in this context is their separation  into two classes that we call {\em trivial} respectively {\em nontrivial} plateaued functions. The former class corresponds to the case when the Walsh support of a plateaued function is an affine subspace. In this case a plateaued function corresponds to a partially bent function and therefore for convenience it is called trivial (see relation (3.1) in \cite{Ayca2}). Furthermore, identifying the signs of nonzero values in the Walsh spectrum of an $s$-plateaued function with any  $(\pm 1)$-bent  sequence on $GF(2)^{n-s}$ gives an efficient design method for trivial plateaued functions. In particular, if our goal is to design a set of disjoint spectra $s$-plateaued functions of maximal cardinality $2^s$ it is sufficient to define these plateaued functions on the cosets of some fixed affine subspace $S_f$ of dimension $n-s$. In difference to previous approaches \cite{Zhang2009,Feng2,Fengrong2018IT} this method is generic, provides a greater variety of the obtained disjoint spectra plateaued functions, and finally the cardinality of these disjoint spectra functions is maximal. We emphasize the fact that any affine subspace $S_f$ and any bent function on $GF(2)^{n-s}$ can be selected in the design, which then essentially exhausts the possibilities of specifying trivial plateaued functions.

To address the problem of EA-equivalence of trivial and nontrivial plateaued functions on $GF(2)^n$, whose Walsh support $S_f \subset \FB^n$ is of cardinality $2^{n-s}$
we use the concept of a {\em dual} of plateaued function, denoted by $f^*$, which essentially gives a correspondence to the (lexicographically) ordered vector space $GF(2)^{n-s}$. It refers to a function associated  to the signs of non-zero Walsh coefficients and is  viewed as a function on $GF(2)^{n-s}$ after fixing a particular ordering on the corresponding Walsh support.
While trivial and nontrivial plateaued functions are always inequivalent, we show that the EA-equivalence between the functions inside these two classes can be completely described in terms of their Walsh supports and duals.

We first completely solved the question related to the design of trivial plateaued functions in the spectral domain and then we investigate the possibilities of designing  nontrivial ones. In this case, the Walsh support is not an affine subspace which then makes a direct correspondence to bent functions on $GF(2)^{n-s}$ more difficult.
The property of being plateaued is then characterized through the bent distance of the dual $f^*$ to so-called {\em sequence profile}  of $S_f$. In difference to trivial plateaued functions the sequence profile of nontrivial plateaued functions necessarily contains sequences of nonlinear functions. In this direction, we give efficient design methods of nontrivial plateaued function which are apparently affine inequivalent to trivial ones. In particular, we also provide an efficient method for constructing nontrivial disjoint spectra plateaued functions of maximal cardinality $2^s$. It is shown that these functions cannot have linear structures which essentially improves upon recent results in \cite{Fengrong2018IT} where large sets of disjoint spectra plateaued functions, but of cardinality less than $2^s$, were specified.

Our analysis leads to two significant research problems that we only address partially.
The first question regards the notion of {\em extendable} set of disjoint spectra $s$-plateaued functions. More precisely, having demonstrated that disjoint spectra $s$-plateaued functions of maximal cardinality can always be specified for trivial plateaued functions and in certain cases for nontrivial ones, the question is whether a given set $\{f_1, \ldots,f_r\}$  of arbitrary nontrivial $s$-plateaued functions can always be extended to a set of disjoint spectra $s$-plateaued functions of maximal cardinality.
Another question which deserves a full attention of the research community is whether the flexibility of defining disjoint spectra $s$-plateaued functions of maximal cardinality might lead to new classes of bent functions. More precisely, using  a set of $2^s$ disjoint spectra of $s$-plateaued functions on $GF(2)^n$ one can easily construct a bent function on the space $GF(2)^{n+s}$ by concatenating these plateaued functions. Then, since the duals of these plateaued functions may be taken from different classes of bent functions it is very challenging to conjecture that the resulting bent function  does not necessarily belong to known primary classes of bent functions. This question is extremely important due to the fact that the primary classes of bent functions only constitute an insignificant
portion of all bent functions.

In the last part of this article we apply a similar method, used to define EA-inequivalent bent functions introduced originally in \cite{HouLang} and later extended in \cite{Houextension2017}, to generate EA-inequivalent plateaued functions by applying a suitable nonlinear permutations of the input space. This approach can also be used to specify nontrivial disjoint spectra plateaued functions but the conditions on the corresponding nonlinear permutations (acting on the input variable space) become rather complicated.

The rest of this article is organized as follows. In Section \ref{sec:pre}, we introduce some basic definitions and the properties of bent and plateaued functions. The EA-equivalence of plateaued functions and their classification into trivial and nontrivial ones is addressed in Section \ref{sec:equivalence}. Furthermore, a necessary and sufficient condition that nontrivial plateaued functions admit linear structures ia also given. In Section \ref{sec:nontrivial}, we provide some generic methods of constructing (non)trivial plateaued functions and show that these methods also gives us the possibility of defining a set of disjoint spectra plateaued functions of maximal cardinality. To generate EA-inequivalent plateaued functions,  we apply in Section \ref{sec:nonlperm} a similar technique of applying suitable nonlinear permutations of the input space, as already employed in \cite{HouLang,Houextension2017}.

\section{Preliminaries}\label{sec:pre}

The vector space $\mathbb{F}_2^n$ is the space of all $n$-tuples $x=(x_1,\ldots,x_n)$, where $x_i \in \mathbb{F}_2$.
For $x=(x_1,\ldots,x_n)$ and $y=(y_1,\ldots,y_n)$  in $\mathbb{F}^n_2$, the usual scalar (or inner) product over $\mathbb{F}_2$ is defined as $x\cdot y=x_1 y_1+\cdots+ x_n y_n.$ The Hamming weight of  $x=(x_1,\ldots,x_n)\in \mathbb{F}^n_2$ is denoted and computed as  $wt(x)=\sum^n_{i=1} x_i.$ 

The set of all Boolean functions in $n$ variables, which is the set of mappings from $\mathbb{F}_2^n$ to $\mathbb{F}_2$, is denoted by $\mathcal{B}_n$.  Especially, the set of affine functions in $n$ variables is given by $\mathcal{A}_n=\{a\cdot x+ b\;|\;a\in\mathbb{F}_2^n,\; b\in\{0,1\}\},$ and similarly  $\mathcal{L}_n=\{a\cdot x:a\in\mathbb{F}_2^n\}\subset \mathcal{A}_n$ 
 denotes the set of linear functions. It is well-known that any $f:\mathbb{F}^n_2 \rightarrow \mathbb{F}_2$ can be uniquely represented by its associated algebraic normal form (ANF) as follows:
\begin{eqnarray}\label{ANF}
f(x_1,\ldots,x_n)={\sum_{u\in \mathbb{F}^n_2}{\lambda_u}}{(\prod_{i=1}^n{x_i}^{u_i})},
\end{eqnarray}
where $x_i, \lambda_u \in \mathbb{F}_2$ and $u=(u_1, \ldots,u_n)\in \mathbb{F}^n_2$. The support of an arbitrary function $f\in \mathcal{B}_n$ is defined as $supp(f)=\{x\in \mathbb{F}^n_2:f(x)=1\}.$ 

For an arbitrary function $f\in \mathcal{B}_n$, the set of its values on $\mathbb{F}^n_2$ (\emph{the truth table}) is defined as $T_f=(f(0,\ldots,0,0),f(0,\ldots,0,1),f(0,\ldots,1,0),\ldots,f(1,\ldots,1,1))$. The corresponding $(\pm 1)$-{\em sequence of $f$} is defined as
$\chi_f=((-1)^{f(0,\ldots,0,0)},(-1)^{f(0,\ldots,0,1)},(-1)^{f(0,\ldots,1,0)}\ldots,$ $(-1)^{f(1,\ldots,1,1)})$. The {\em Hamming distance} $d_H$ between two arbitrary Boolean functions, say $f,g\in \mathcal{B}_n,$ we define by $d_H(f,g)=\{x\in \mathbb{F}^n_2:f(x)\neq g(x)\}=2^{n-1}-\frac{1}{2}\chi_f\cdot \chi_g$, where $\chi_f\cdot \chi_g=\sum_{x\in \mathbb{F}^n_2}(-1)^{f(x)+ g(x)}$.
By $\textbf{0}_k=(0,\ldots,0)$ we denote the all-zero vector in $\mathbb{F}^k_2.$

The \emph{Walsh-Hadamard transform} (WHT) of $f\in\mathcal{B}_n$, and its inverse WHT, at any point $\omega\in\mathbb{F}^n_2$ are defined respectively by
\begin{eqnarray}\label{WHT}
W_{f}(\omega)=\sum_{x\in \mathbb{F}_2^n}(-1)^{f(x)+ \omega\cdot x},\;\;\;\;(-1)^{f(x)}=2^{-n}\sum_{\omega\in \mathbb{F}_2^n}W_f(\omega)(-1)^{\omega\cdot x}, \;\;\;\omega\in \mathbb{F}^n_2.
\end{eqnarray}

For a set $U\subseteq \mathbb{F}^n_2$, by $U^{\bot}$ we denote the set  $U^{\bot}=\{y\in \mathbb{F}^n_2:x\cdot y=0,\;\forall x\in U\}.$ The cardinality of any set $U$ is  denoted by  $\#U.$ The \emph{Sylvester-Hadamard}  matrix of size $2^k \times 2^k$, is defined recursively as:
\begin{eqnarray*}\label{HM}
H_1=(1);\hskip 0.4cm H_2=\left(
                           \begin{array}{cc}
                             1 & 1 \\
                             1 & -1 \\
                           \end{array}
                         \right);\hskip 0.4cm H_{2^k}=\left(
      \begin{array}{cc}
        H_{2^{k-1}} & H_{2^{k-1}} \\
        H_{2^{k-1}} & -H_{2^{k-1}} \\
      \end{array}
    \right).
\end{eqnarray*}
The $i$-th row of $H_{2^k}$ we denote by $H^{(i)}_{2^k}$ ($i\in[0,2^k-1]$). For any two sets $A=\{\alpha_1,\ldots,\alpha_{r}\}$ and $B=\{\beta_1,\ldots,\beta_{r}\}$, we define  $A\wr B=\{(\alpha_i,\beta_i): i=1,\ldots,r\}$.

\subsection{Bent and plateaued functions and their duals}

Throughout this article we use the following definitions and useful facts related to bent and plateaued functions.
 A function $f\in\mathcal{B}_n,$ for even $n$, is called {\em bent} if $W_f(u)=2^{\frac{n}{2}}(-1)^{f^*(u)}$
for a Boolean function $f^*\in \mathcal{B}_n$ which is also a bent function, called the {\it dual} of $f$.
Two functions $f$ and $g$ on $\FB^n$ are said to be at {\em bent distance} if $d_H(f,g)= 2^{n-1}\pm 2^{n/2-1}$. Similarly, for a subset $B\subset \mathcal{B}_n$, a function $f$ is said to be at bent distance to $B$ if for all $g\in B$ it holds that $d_H(f,g)=2^{n-1}\pm 2^{n/2-1}$.

A function $f\in \mathcal{B}_n$ is called {\em $s$-plateaued} if its Walsh spectrum only takes three values $0$ and $\pm 2^{\frac{n+s}{2}}$ ($\leq 2^n$), where $s\geq 1$ if $n$ is odd and $s\geq 2$ if $n$ is even ($s$ and $n$ always have the same parity). The Walsh distribution of $s$-plateaued functions (cf. \cite[Proposition 4]{Decom}) is given by
\begin{table}[H]
\caption{Walsh spectra distribution}
\label{tab:plateaued}
\begin{eqnarray*}
\begin{tabular}{|c|c|}
  \hline
  $W_f(u)$ & Number of $u\in \mathbb{F}^n_2$ \\ \hline
  $0$ & $2^n-2^{n-s}$ \\ \hline
  $2^{\frac{n+s}{2}}$ & $2^{n-s-1}+(-1)^{f(0)}2^{\frac{n-s}{2}-1}$ \\ \hline
  $-2^{\frac{n+s}{2}}$ & $2^{n-s-1}-(-1)^{f(0)}2^{\frac{n-s}{2}-1}$ \\
  \hline
\end{tabular}
\end{eqnarray*}
\end{table}
In particular, a class of $1$-plateaued functions for $n$ odd, or  $2$-plateaued for $n$ even, corresponds to so-called {\em semi-bent} functions.
In general, the {\em Walsh support} of  $f\in \mathcal{B}_n$ is defined as $S_f=\{\omega\in \mathbb{F}^n_2\; :\; W_f(\omega)\neq0\}.$

 For an arbitrary  $s$-plateaued function $f\in \mathcal{B}_n$ with $W_f(u)\in \{0,\pm 2^{\frac{n+s}{2}}\}$ the value $2^{\frac{n+s}{2}}$ is called the \emph{amplitude} of $f$. We define its dual function $f^*$ on the set $S_f$ of cardinality $2^{n-s}$ by $W_f(\omega)=2^{\frac{n+s}{2}}(-1)^{f^*(\omega)},$ for $\omega\in S_f$.
To specify the dual function as $f^*:\FB^{n-s} \rightarrow \FB$ we use the concept of {\em lexicographic ordering}. That is, a subset  $E=\{e_0,\ldots,e_{2^{n-s}-1}\}\subset \mathbb{F}^{n}_2$ is ordered lexicographically if $|e_i| < |e_{i+1}|$ for any $i \in [0,2^{n-s}-2]$, where $|e_i|$ denotes the integer representation of $e_i \in \mathbb{F}^n_2$. More precisely, for $e_i=(e_{i,0}, \ldots, e_{i,n-1})$ we have $|e_i|= \sum_{j=0}^{n-1}e_{i,n-1-j} 2^j$, thus having the most significant bit of $e_i$ on the left-hand side.
Since $S_f$ is not ordered in general, {\em we will always represent it} as $S_f=v + EM$,\footnote{The reason for embedding $M$ to describe $S_f$ and not use simply $S_f=v +E$ will be explained in Section \ref{sec:class}, in the context of EA-equivalence of plateaued functions.} for some lexicographically ordered set $E$ with $e_0=\textbf{0}_{n}$, $M\in GL(n,\mathbb{F}_2)$ and $v \in S_f$. Here, $GL(n,\mathbb{F}_2)$ denotes the group of all invertible $\mathbb{F}_2$-linear transformations on $\mathbb{F}^n_2$, i.e., the group of invertible binary matrices of size $n \times n$.

For instance, if $S_f=\{(0,1,0),(0,1,1),(1,1,1), (1,0,1)\}$, by fixing $v=(0,1,1)\in S_f$ and $M$ to be the identity matrix $I_{3\times 3}$, then $E=\{e_0,e_1,e_2,e_3\}=\{(0,0,0), (0,0,1),(1,0,0),(1,1,0)\}$ is ordered lexicographically and consequently $S_f=v+ E$ is "ordered" as $S_f=\{\omega_0,\omega_1,\omega_2,\omega_3\}=\{(0,1,1),(0,1,0),(1,1,1),(1,0,1)\}$. Notice that in this example $S_f$ is not an affine subspace.

 This way we can make a direct correspondence between  $\FB^{n-s}$ and $S_f$ through $E$ so that for $\FB^{n-s}=\{x_0, x_1, \ldots, x_{2^{n-s}-1}\}$, where $\FB^{n-s}$ is lexicographically ordered, we assign
\begin{eqnarray}\label{DPL}
 \overline{f}^*(x_j)=f^*(v+e_jM)=f^*(\omega_j),\;\;\; x_j \in \FB^{n-s}, \; e_j \in E,\;j\in[0,2^{n-s}-1].
 \end{eqnarray}
Then, the set $S_f=\{\omega_0,\ldots,\omega_{2^{n-s}-1}\}$ is "ordered" so that $\omega_i=v+ e_i M$, where $E=\{e_0,\ldots,e_{2^{n-s}-1}\}$ is ordered lexicographically.  In the above example, we have for instance that $\overline{f}^*(x_1)=\overline{f}^*(0,1)= f^*(0,1,0)=f^*(\omega_1)$, $x_1\in \mathbb{F}^2_2$.
This simply means that the truth table of $\overline{f}^*:\FB^{n-s} \rightarrow \FB$ is fully specified by assigning its  values with respect to   the signs of the Walsh spectral values in the ordered set $S_f=\{\omega_0,\omega_1,\ldots,\omega_{2^{n-s}-1}\}$.
\begin{rem}\label{rem:dual}
This definition of  $\overline{f}^*$ (given by (\ref{DPL})), when viewed as a function on $\FB^{n-s}$ is not unique, since the given set $S_f \subset \FB^n$ can be represented in many ways as $S_f=v+ EM$, for some $v \in \FB^n$, $M \in GL(n,\mathbb{F}_2)$ and lexicographically ordered set $E$.  This means that the same function $f$ gives rise to different duals defined on $\mathbb{F}^{n-s}_2$ through the relation (\ref{DPL}).
An alternative way of describing the relation (\ref{DPL}) would be the use of a mapping, say $P:S_f\rightarrow \mathbb{F}^{n-s}_2$, and write (\ref{DPL}) as $f^*(\omega_j)=f^*(P(\omega_j))=\overline{f}^*(x_j)$ (where $f^*\circ P$ is viewed as $\overline{f}^*$ defined on $\mathbb{F}^{n-s}_2$), but its properties are not known due to the  translation/preservation of an ordering imposed on $S_f$ to $\mathbb{F}^{n-s}_2$ (see \cite[Remark 2.1]{Secondary}). Note that when $S_f$ is an affine subspace, then the mapping $P$ can be described in terms of (\ref{eq:Q}), but if $S_f$ is non-affine its description is not known.
\end{rem}

\section{Classification of plateaued functions}\label{sec:class}

In this section we discuss different ways to classify plateaued functions depending on their spectral characterization. In the first place, we observe that the standard notion of extended affine (EA) equivalence transfers to the spectral domain which enables us to distinguish EA-inequivalent classes of  plateaued functions. A plateaued function with affine Walsh support is then called trivial and its equivalence class is obtained by applying invertible affine transformation in the spectral domain. In difference to these trivial objects (that correspond to partially bent functions \cite[Relation (3.1)]{Ayca2}) there are nontrivial plateaued functions with non-affine Walsh support that are EA-inequivalent to any trivial plateaued function. By analyzing the EA-equivalence through relation (\ref{DPL}), we provide both description and generic construction of inequivalent plateaued functions, regardless of whether they are trivial or not. We also show that certain subclasses of nontrivial plateaued functions functions do not admit linear structures.

\subsection{On affine equivalence between (non)trivial plateaued functions}\label{sec:equivalence}

Since the design of plateaued functions described in \cite[Section 3]{Secondary} is related to the Walsh spectrum, it is convenient to analyze the EA-equivalence between two arbitrary plateaued functions in terms of their Walsh supports and duals. The  notion of EA-equivalence, which goes back to early works of Marioana-McFarland \cite{MMclass} and Dillon \cite{Dillon} can be stated as follows. Throughout this article $GL(n,\mathbb{F}_2)$ denotes the group of all invertible $\mathbb{F}_2$-linear transformations on $\mathbb{F}^n_2$.
\begin{defi}\label{def:EA}
Two Boolean functions $h,f:\mathbb{F}^n_2\rightarrow \mathbb{F}_2$ are said to be EA-equivalent if there exists a matrix $A\in GL(n,\mathbb{F}_2)$,
vectors $b,c\in \mathbb{F}^n_2$ and a constant $\varepsilon \in \{0,1\}$ such that
\begin{equation} \label{eq:EAeq}
h(x)=f(xA+b)+c\cdot x+\varepsilon.
\end{equation}
\end{defi}
This equivalence naturally applies to plateaued functions and if $f$ is an $s$-plateaued function on $\FB^n$ with the support $S_f$ then the Walsh spectra of $h(x)=f(xA+b)+c\cdot x+\varepsilon$ (in terms of relation (\ref{DPL})), at any $u\in \mathbb{F}^n_2$, is given by:
\begin{eqnarray}\label{eq:equiv}\nonumber
W_h(u)&=&\sum_{x\in \mathbb{F}^n_2}(-1)^{f(xA+b)+c\cdot x+\varepsilon+u\cdot x}=\sum_{y\in \mathbb{F}^n_2}(-1)^{f(y)+(u+c)A^{-T}\cdot y+(u+c)\cdot bA^{-1}+\varepsilon}\\
&=&\left\{\begin{array}{cc}
                                                                                                 2^{\frac{n+s}{2}}(-1)^{f^*((u+c)A^{-T})+(u+c)\cdot bA^{-1}+\varepsilon}=2^{\frac{n+s}{2}}(-1)^{h^*(u)}, & (u+c)A^{-T}\in S_f \\
                                                                                                 0, &  (u+c)A^{-T}\not\in S_f
                                                                                               \end{array}
\right..
\end{eqnarray}
From (\ref{eq:equiv}), we have that the Walsh supports of $h$ and $f$ are related by the equality $S_h=c+S_fA^T$. In order to establish a more precise connection between the duals of two EA-equivalent $s$-plateaued functions, say $f$ and  $h$ on $\mathbb{F}^n_2$, we fix the following notation.
\begin{itemize}
 \item We  always assume that  $S_f=v+E=\{\omega_0,\ldots,\omega_{2^{n-s}-1}\}$ is  ordered so that $\omega_i=v+e_i$, where $v\in S_f$ and $E=\{e_0,\ldots,e_{2^{n-s}-1}\}$ is ordered lexicographically ($e_0=\textbf{0}_n$).
\item  Since by (\ref{eq:equiv}) we have $S_h=c+S_fA^T$, the support of $h$ is  ordered as $S_h=\{z_0,\ldots,z_{2^{n-s}-1}\}$ with $z_i=c+\omega_iA^T$ $(i\in[0,2^{n-s}-1])$. With respect to  $S_f$ and $S_h$, the duals $\overline{f}^*,\overline{h}^*:\mathbb{F}^{n-s}_2\rightarrow \mathbb{F}_2$ (where e.g. $x_i \rightarrow \overline{f}^*(x_i$)) will be defined using the identification:
\begin{equation}
\label{eq:identif}
\overline{f}^*(x_i)= f^*(\omega_i);  \;\;\;\;\;\;\;\;  \overline{h}^*(x_i)= h^*(z_i); \; \; \;\; x_i \in \FB^{n-s}, \omega_i \in S_f, z_i \in S_h,
\end{equation}
where in particular (by relations (\ref{eq:equiv}) and (\ref{eq:identif})) we have that
 \begin{eqnarray}\label{eq:equiv22}
h^*(z_i)=f^*(\omega_i)+\omega_i\cdot b+\varepsilon. 
\end{eqnarray}
where $z_i=c+\omega_iA^T$  and $i\in[0,2^{n-s}-1]$.
\item
Using $S_f=v+E$ (with lexicographically ordered $E$) and expressing   $S_h=c+S_fA^T=\widehat{v}+EA^T$, where $\widehat{v}=c+vA^T\in \mathbb{F}^n_2$,  gives us the possibility to analyze the EA-equivalence in a sensible way by representing $S_h=\widehat{v} + EM$, where we have that $M=A^T$.
\end{itemize}
%
\begin{rem}\label{rem:ordering}
In relation (\ref{eq:equiv22}), if  $S_f$ is an ordered set then it  induces some  ordering of $S_h$ via the equality $z_i=c+\omega_iA^T$, and clearly the Walsh supports of any two equivalent plateaued functions $h$ and $f$ are then related using this connection. 
\end{rem}
In general, the equivalence between $h$ and $f$ is possible regardless of whether $S_h$ and $S_f$ are affine subspaces or not. 
\begin{defi} \label{def:nontr}
An $s$-plateaued function $f \in \mathcal{B}_n$ whose Walsh support $S_f$ is an affine subspace is called {\em trivial}, otherwise  $f$ is  called {\em nontrivial} if $S_f$ is not an affine subspace.
\end{defi}
%
%
In order to proceed further with our analysis, we recall some  notation introduced in \cite{Secondary}. More precisely, for an arbitrary $s$-plateaued function $f$ defined on $\mathbb{F}^n_2$ let $S_f\subset \mathbb{F}^n_2$ ($\#S_f=2^{n-s}$) be its Walsh support ordered as $S_f=\{\omega_{0},\ldots,\omega_{2^{n-s}-1}\}$. The \emph{sequence profile} of $S_f$,
 which is a multi-set of $2^n$ sequences of length $2^{n-s}$ induced by $S_f$, is defined as 
\begin{eqnarray}\label{RSf}
\Phi_{f}=\{\phi_u:\mathbb{F}^{n-s}_2\rightarrow \mathbb{F}_2\;:\; \chi_{\phi_u}=((-1)^{u\cdot \omega_{0}},(-1)^{u\cdot \omega_{1}},\ldots,(-1)^{u\cdot \omega_{2^{n-s}-1}}),\;\omega_i\in S_f,\; u\in \mathbb{F}^{n}_2\}.
\end{eqnarray}
As noted in \cite{Secondary}, $\Phi_{f}$ depends on the ordering of $S_f$ and it is spanned by the functions $\phi_{b_1},\ldots,\phi_{b_n}$, i.e., $\Phi_{f}=\langle \phi_{b_1},\ldots,\phi_{b_n}\rangle$, where $b_1,\ldots,b_n$ is the canonical basis of $\mathbb{F}^n_2$ ($b_i$ contains the non-zero coordinate at the $i$-th position).

The following result  combines the relations (\ref{eq:identif}), (\ref{eq:equiv22}) and (\ref{RSf}) to specify EA-(in)equivalence between two plateaued functions.
\begin{theo}\label{th:eq}
Let $f,h:\mathbb{F}^n_2\rightarrow \mathbb{F}_2$ be $s$-plateaued functions whose Walsh supports are given as $S_f=\{\omega_0,\ldots,\omega_{2^{n-s}-1}\}$ and $S_h=\{z_0,\ldots,z_{2^{n-s}-1}\}$. Suppose that the duals $\overline{f}^*, \overline{h}^*:\mathbb{F}^{n-s}_2\rightarrow \mathbb{F}_2$ are defined using (\ref{eq:identif}).
Then:
\begin{enumerate}[i)]
\item If there does not exist a vector $c\in \mathbb{F}^n_2$ and a matrix $A\in GL(n,\mathbb{F}_2)$ such that $S_h=c+S_fA^T$, then $f$ and $h$ are EA-inequivalent.
\item $f$ and $h$ are EA-equivalent (in terms of (\ref{eq:EAeq})) if and only if there exist  $c,b\in \mathbb{F}^n_2$, $A\in GL(n,\mathbb{F}_2)$ and $\varepsilon \in \mathbb{F}_2$ such that $S_h=c+S_fA^T$ $($thus $z_i=c+\omega_iA^T)$ and $\overline{h}^*(x_i)=\overline{f}^*(x_i)+\phi_{b}(x_i)+\varepsilon$, for all $i \in [0,2^{n-s}-1]$. 
\end{enumerate}
\end{theo}
\proof While $i)$ follows trivially from (\ref{eq:equiv}), we focus on $ii)$. From (\ref{eq:equiv}), we have that
$$h^*(z_i)=f^*((z_i+c)A^{-T})+(z_i+c)\cdot bA^{-1}+\varepsilon,\;\;\;\omega_i=(z_i+c)A^{-T}\in S_f.$$
Since  $z_i\in S_h$, then from $(z_i+c)A^{-T}\in S_f$ we can write $z_i=c+\omega_iA^T$, for some (unique) $\omega_i\in S_f$. Thus, the equality above can be written as
\begin{eqnarray*}
h^*(z_i)=h^*(c+\omega_iA^T)&=&f^*(\omega_i)+(c+\omega_iA^T+c)\cdot bA^{-1}+\varepsilon\\
&=&f^*(\omega_i)+\omega_iA^T\cdot bA^{-1}+\varepsilon=f^*(\omega_i)+\omega_i\cdot bA^{-1}A+\varepsilon\\
&=&f^*(\omega_i)+\omega_i\cdot b+\varepsilon\stackrel{(\ref{RSf})}{=}f^*(\omega_i)+\phi_b(\omega_i)+\varepsilon,
\end{eqnarray*}
which completes the proof due to relations (\ref{eq:identif}), (\ref{eq:equiv22}) and (\ref{RSf}).\qed  
%
%
By Theorem \ref{th:eq}-$ i)$, we clearly have that trivial and nontrivial $s$-plateaued functions in $n$ variables cannot be EA-equivalent. In the case of nontrivial plateaued functions, we notice that their Walsh supports in general may contain different number of linearly independent vectors. This fact can be used as  an efficient indicator for EA-inequivalence  since affine mappings preserve the cardinality of linearly independent vectors and the equality $S_h=c+S_fA^T$ will not  always be satisfied (thus we can distinguish EA-inequivalent functions $f$ and $h$). Nevertheless, in the case when this equality holds (which also applies to trivial plateaued functions), then the EA-equivalence of $f$ and $h$  strictly depends on the relation between its duals $\overline{h}^*$ and $\overline{f}^*$, which is given in Theorem \ref{th:eq}-$(ii)$ as $\overline{h}^*=\overline{f}^*+\phi_b+\varepsilon$.

Now we focus on the analysis of trivial plateaued functions. We firstly provide the following result which slightly extends \cite[Lemma 3.1]{Secondary} by involving invertible matrices, where for conciseness we omit mentioning that $E$ is lexicographically ordered.
\begin{prop}\label{prop:Hrow}
Let $S=\{\omega_0,\ldots,\omega_{2^m-1}\}\subseteq\mathbb{F}^n_2$ be an affine subspace represented as $S=v+E$, where $E=\{e_0,\ldots,e_{2^m-1}\}\subset \mathbb{F}^n_2$ is a linear subspace.
Let for an arbitrary  $M\in GL(n,\mathbb{F}_2)$ and $c\in \mathbb{F}^n_2$ the set $\widetilde{S}$ be given by $\widetilde{S}=c+SM=\{\alpha_0,\ldots,\alpha_{2^m-1}\}$ ($\alpha_i=c+(v+e_i)M$, $i\in[0,2^m-1]$). Then, for any $u\in \mathbb{F}^n_2$, it holds that
$$((-1)^{u\cdot \alpha_0},\ldots, (-1)^{u\cdot \alpha_{2^m-1}})=(-1)^{\varepsilon_u}H^{(r_u)}_{2^{m}},$$
for some $0\leq r_u\leq 2^m-1$ and $\varepsilon_u\in \mathbb{F}_2$. In addition,  it holds that $\{T_{\ell}:\ell\in \mathcal{L}_m\}\subseteq \{(u\cdot e_0,\ldots,u\cdot e_{2^m-1}):u\in \mathbb{F}^n_2\}$. 
\end{prop}
\begin{proof} Denoting by $c'=c+vM\in \mathbb{F}^n_2$, we have that $\alpha_i=c+(v+e_i)M=c'+e_iM$ for all $i\in[0,2^m-1]$. Consequently, for any $u\in \mathbb{F}^n_2$, we have that $u\cdot \alpha_i=u\cdot c'+u\cdot e_iM=u\cdot c'+uM^T\cdot e_i$, and thus by \cite[Lemma 3.1]{Secondary} it holds that
$$((-1)^{u\cdot \alpha_0},\ldots, (-1)^{u\cdot \alpha_{2^m-1}})=(-1)^{u\cdot c'}((-1)^{uM^T\cdot e_0},\ldots,(-1)^{uM^T\cdot e_{2^m-1}})=(-1)^{\varepsilon_u}H^{(r_u)}_{2^m},$$
for some $0\leq r_u\leq 2^m-1$ and $\varepsilon_u\in \mathbb{F}_2$. The rest of the statement follows by \cite[Lemma 3.1]{Secondary}.\qed
\end{proof}
Suppose that a linear subspace $E=\{e_0,\ldots,e_{2^{n-s}-1}\}\subset {\Bbb F}_2^n$ ($\dim(E)=n-s$) and ${\Bbb F}_2^{n-s}=\{x_0,x_1,x_2,\ldots,x_{2^{n-s}-1}\}$ are ordered lexicographically.  From the proof of Lemma 3.1-$(ii)$ given in \cite{Secondary}, the linear mapping $\Theta(x_i)=e_i$, with $\Theta: {\Bbb F}_2^{n-s} \rightarrow E$,  can be described via a matrix $R=R_{(n-s)\times n}$ so that
\begin{eqnarray}\label{eq:Q}
\Theta(x_i)=x_iR=x_i\left(
                                                                                        \begin{array}{c}
                                                                                          e_{2^{n-s-1}} \\
                                                                                          \vdots \\
                                                                                          e_2 \\
                                                                                          e_1 \\
                                                                                        \end{array}
                                                                                      \right)=e_i,\;\;\;\;i\in[0,2^{n-s}-1],
\end{eqnarray}
where the rows of  $R$, $e_{2^j}$ for $j =0,\ldots, n-s-1$, are basis vectors of $E$. We have the following technical result useful for establishing  EA-equivalence of plateaued functions.
%
\begin{lemma}\label{prop:Hrow1}
Let a linear subspace $E\subset {\Bbb F}_2^n$ and ${\Bbb F}_2^{n-s}$ be lexicographically ordered, and let the mapping $\Theta:{\Bbb F}_2^{n-s} \rightarrow E$ be defined by (\ref{eq:Q}). Then for an arbitrary (fixed) matrix $B\in GL(n-s,\mathbb{F}_2)$ and vector $t\in \mathbb{F}^{n-s}_2$, there exists a vector $\gamma\in E$ and a matrix $D\in GL(n,\mathbb{F}_2)$ which is invariant on $E$ (that is $E=ED$) such that
 $$\Theta(x_iB+t)=e_iD+\gamma,$$
 holds for all $i\in[0,2^{n-s}-1]$.
\end{lemma}
\begin{proof}
Since the rows of  $R$ are basis vectors $e_1,e_2,e_4,\ldots,e_{2^{n-s-1}}\in E$, then clearly we can always find a set of vectors $\lambda_1,\ldots,\lambda_s\in \mathbb{F}^n_2$ such that $\{e_1,\ldots,e_{2^{n-s-1}},\lambda_1,\ldots,\lambda_s\}$ is a basis of $\mathbb{F}^n_2$. Now, let $\Lambda=\Lambda_{s\times n}=(\lambda_1,\ldots,\lambda_s)^T$ denotes the matrix whose rows are $\lambda_i$ vectors. Then, for a matrix $B\in GL(n-s,\mathbb{F}_2)$ and an arbitrary $x_i\in \mathbb{F}^{n-s}_2$ it holds that
\begin{eqnarray}\label{eq:ED}
\Theta(x_iB)=(x_iB)R=(x_i,\textbf{0}_s)\left(
                                           \begin{array}{c}
                                             BR \\
                                             \Lambda \\
                                           \end{array}
                                         \right)_{n\times n}=(x_i,\textbf{0}_s)Q=e_i\in E,
\end{eqnarray}
since $\tilde{x}_i=x_iB\in \mathbb{F}^{n-s}_2$ and $\Theta(\mathbb{F}^{n-s}_2)\in E$. Here, the block matrix $Q=\left(
                                           \begin{array}{c}
                                             BR \\
                                             \Lambda \\
                                           \end{array}
                                         \right)_{n\times n}$ is clearly invertible, i.e., $Q\in GL(n,\mathbb{F}_2)$, since the matrix $BR$ (of size $(n-s)\times n$) contains rows which are linear combinations of vectors $e_{2^j}$, $j\in[1,n-s-1]$ (which are linearly independent of $\lambda_i$).
Moreover, denoting by $U=\left(
                                           \begin{array}{c}
                                             R \\
                                             \Lambda \\
                                           \end{array}
                                         \right)_{n\times n}\in GL(n,\mathbb{F}_2)$, we have that
$$\Theta(x_i)=x_iR=(x_i,\textbf{0}_s)U=e_i,$$ implies that $(x_i,\textbf{0}_s)=e_iU^{-1}$ holds for all $i\in[0,2^{n-s}-1].$ Consequently, for an arbitrary vector $t\in \mathbb{F}^{n-s}_2$ we have that
\begin{eqnarray*}
\Theta(x_iB+t)=(x_iB+t)R=x_i(BR)+tR=(x_i,\textbf{0}_s)Q+tR=e_iU^{-1}Q+tR=e_iD+\gamma\in E,
\end{eqnarray*}
where $D=U^{-1}Q\in GL(n,\mathbb{F}_2)$ and $tR=\gamma\in E$ due to the facts that $\Theta(x_iB)=x_i(BR)=e_iD\in E$ holds by (\ref{eq:ED}). Also, $\Theta(x_iB+t)\in \Theta(\mathbb{F}^{n-s}_2)=E$ since $E$ is a linear subspace.\qed
\end{proof}
Recall that any two EA-equivalent trivial plateaued functions $f$ and $h$ have the Walsh supports related in an affine manner as $S_h=c+S_fM$ (for some invertible matrix $M$ and vector $c$). If there do not exist $M$ and $c$ such that $S_h=c+S_fM$ holds, then by Theorem \ref{th:eq}-$(i)$ (or relation (\ref{eq:equiv})) the functions $f$ and $h$ are not equivalent. In other words, the EA-equivalence between $f$ and $h$ is reasonable to consider only under the assumption that $S_h$ and $S_f$ are related in an affine manner.
The following result characterizes the EA-equivalence between trivial plateaued functions.
\begin{theo}\label{th:EA1}
Let $f,h:\mathbb{F}^n_2\rightarrow \mathbb{F}_2$ be two trivial $s$-plateaued functions whose Walsh supports are related as $S_h=c+S_fM$, for some matrix $M\in GL(n,\mathbb{F}_2)$ and $c\in \mathbb{F}^n_2.$ Representing $S_f=v+E=\{\omega_i=v+e_i:e_i\in E\}$ for a lexicographically ordered linear space $E=\{e_0,\ldots,e_{2^{n-s}-1}\}$, let the functions $\overline{f}^*$ and $\overline{h}^*$ be defined as
$$\overline{f}^*(x_i)=f^*(\omega_i)\;\;\;\text{and}\;\;\;\overline{h}^*(x_i)=h^*(z_i),\;\;\;(i\in[0,2^{n-s}-1]),$$
where $z_i=c+\omega_iM\in S_h$. Then $f$ and $h$ are EA-equivalent if and only if their duals $\overline{f}^*,\overline{h}^*:\mathbb{F}^{n-s}_2\rightarrow\mathbb{F}_2$ 
are EA-equivalent bent functions. 
\end{theo}
\begin{proof} $(\Rightarrow)$ By \cite[Theorem 3.1-$(ii)$]{Secondary} and relations (\ref{eq:identif})-(\ref{eq:equiv22}) it is clear that $\overline{f}^*,\overline{h}^*:\mathbb{F}^{n-s}_2\rightarrow\mathbb{F}_2$  are EA-equivalent bent functions.

$(\Leftarrow)$ In this part, we assume that $f$ and $h$ are two fixed trivial plateaued functions and thus the values of $h^*$ (on $S_h$) and $f^*$ (on $S_f$) are fixed in advance.
%
Due to Proposition \ref{prop:Hrow} and Theorem \ref{theo:plateH}-$(ii)$ (which is given later on in Section \ref{sec:ineq}), it is not difficult to see that both duals $\overline{h}^*$ and $\overline{f}^*$ are bent functions.

Now, let us assume that $\overline{f}^*,\overline{h}^*$  are affine equivalent, that is for some $B\in GL(n-s,\mathbb{F}_2)$, $t,r\in \mathbb{F}^{n-s}_2$ and $\kappa\in \mathbb{F}_2$, we have
\begin{eqnarray}\label{eq:w0}
\overline{h}^*(x_i)=\overline{f}^*(x_iB+t)+r\cdot x_i+\kappa, \;\;\;\;\; x_i\in \mathbb{F}^{n-s}_2.
\end{eqnarray}
%
%
Since $S_f$ is an affine subspace, then by Proposition \ref{prop:Hrow1} there exist  $b\in \mathbb{F}^n_2$ and $\varepsilon\in \mathbb{F}_2$ such that
\begin{eqnarray}\label{eq:w1}
b\cdot \omega_i+\varepsilon=r\cdot x_i+\kappa,
\end{eqnarray}
where $(b\cdot e_0,\ldots,b\cdot e_{2^{n-s}-1})=(r\cdot x_0,\ldots,r\cdot x_{2^{n-s}-1})$ ($x_i\in \mathbb{F}^{n-s}_2$) is truth table of a linear function in $n-s$ variables, and $b\cdot v+\varepsilon=\kappa$.

Furthermore, since by Lemma \ref{prop:Hrow1} we have that $\Theta(x_iB+t)=e_iD+\gamma\in E$, then clearly we can identify vectors $x_iB+t\in \mathbb{F}^{n-s}_2$ with vectors $v+(e_iD+\gamma)\in v+E=S_f$, and thus we have that
\begin{eqnarray}\label{eq:w2}
\overline{f}^*(x_iB+t)=f^*(v+(e_iD+\gamma)),\;\;\;\;i\in[0,2^{n-s}-1].
\end{eqnarray}
This is possible due to the fact that the definition of $\overline{f}^*$ on $\mathbb{F}^{n-s}_2$ (given by (\ref{DPL})) is strictly given with respect to vectors $e_i\in E$ in $f^*(v+e_i)$, where the vector $v$ is ignored.
%
Combining the previous computation with Lemma \ref{prop:Hrow1} and relations (\ref{eq:w1})-(\ref{eq:w2}), we have that
\begin{eqnarray}\label{eq:w3}\nonumber
\overline{f}^*(x_iB+t)+r\cdot x_i+\kappa \hskip -2mm &=&\hskip -2mm f^*(v+(e_iD+\gamma))+b\cdot \omega_i+\varepsilon=f^*(v+\gamma+vD+\omega_iD)+b\cdot \omega_i+\varepsilon\\
&=&f^*(t_i)+b\cdot (t_i+v+\gamma+vD)D^{-1}+\varepsilon,
\end{eqnarray}
where we have used the substitution $t_i=v+\gamma+vD+\omega_iD\in S_f$ which gives that $\omega_i=(t_i+v+\gamma+vD)D^{-1}.$

Consequently, by relations (\ref{eq:w0})-(\ref{eq:w3}), for arbitrary $x\in \mathbb{F}^n_2$ it holds that
 \begin{equation*}
 \begin{array}{rl}
 (-1)^{h(x)}=&2^{-n}\sum\limits_{z_i\in S_h}W_h(z_i)(-1)^{x\cdot z_i}=2^{\frac{s-n}{2}}\sum\limits_{z_i\in S_h}(-1)^{h^*(z_i)+x\cdot z_i}\\
=&2^{\frac{s-n}{2}}\sum\limits_{e_i\in E}(-1)^{f^*(v+e_iD+\gamma)+b\cdot (t_i+v+\gamma+vD)D^{-1}+x\cdot (c+vM+e_iM)+\varepsilon}\\
\stackrel{t_i=v+e_iD+\gamma}{=}&2^{\frac{s-n}{2}}\sum\limits_{t_i\in S_f}(-1)^{f^*(t_i)+b\cdot (t_i+v+\gamma+vD)D^{-1}+x\cdot [c+(t_i+v+\gamma+vD)D^{-1}M]+\varepsilon}\\
=&2^{-n}\sum\limits_{t_i\in S_f}W_f(t_i)(-1)^{t_i\cdot [bD^{-T}+x(D^{-1}M)^{T}]+ b\cdot (v+\gamma+vD)D^{-1}+ x\cdot [c+(v+\gamma+vD)D^{-1}M] +\varepsilon}\\
{=}& (-1)^{f(bD^{-T}+x(D^{-1}M)^{T})+x\cdot [c+(v+\gamma+vD)D^{-1}M]+[b\cdot (v+\gamma+vD)D^{-1} +\varepsilon]},\\
 \end{array}
 \end{equation*}
which means that $f$ and $h$ are affine equivalent.\qed

\begin{rem}\label{rem:ordering}
Note that Theorem \ref{th:EA1} also holds  when the set $E$ is ordered as $E=\{\hat{e}_0,\hat{e}_1,\ldots,\hat{e}_{2^{n-s}-1}\}$, where the vectors $\hat{e}_j$, for any $i\in[0,n-s-1]$, satisfy the recursion $\hat{e}_j=\hat{e}_{2^i}+\hat{e}_{j-2^i}$ for all $2^i\leq j\leq 2^{i+1},$ in which case the set $\{\hat{e}_{2^i}:i\in[1,n-s-1]\}$ is a basis of $E$. Regarding the equality $S_h=c+S_fM$, it is clear that the corresponding ordering of $S_h=\{\hat{z}_0,\ldots,\hat{z}_{2^{n-s}-1}\}$ with $\hat{z}_j=c+(v+\hat{e}_j)M$ (where $v+\hat{e}_j\in S_f$) satisfies the same recursion $\hat{z}_j=\hat{z}_{2^i}+\hat{z}_{j-2^i}$. Thus, by Lemma 3.1-$(i)$ and Theorem 3.1 given in \cite{Secondary} (or Theorem \ref{theo:plateH}),  both duals $\overline{f}^*(x_i)=f^*(v+\hat{e}_j)$ and $\overline{h}^*(x_i)=f^*(\hat{z}_j)$ are still bent functions on $\mathbb{F}^{n-s}_2$. Another proof in this context (without using a restriction on lexicographic ordering of $E$) is given in Appendix.
\end{rem}
Note that second part of the proof of Theorem \ref{th:EA1} embeds a spectral equivalence which is used in Section \ref{sec:ineq} (cf. Theorem \ref{theo:plateH}) to provide an efficient method of constructing inequivalent trivial plateaued functions through suitably  specified duals.
%
\end{proof}

\subsection{Generic construction of EA-inequivalent plateaued functions}\label{sec:ineq}

A generic construction method of plateaued functions, which specifies these functions in their Walsh spectra, has been introduced in \cite[Section 3]{Secondary}. In order to distinguish it from the direct or shortly \emph{ANF-construction methods}\footnote{An  \emph{ANF-construction method} refers to any construction technique which specifies a form or ANF of the function, for which it can be shown to posses the plateaued property.}, the construction technique given in \cite[Section 3]{Secondary} will be called a \emph{spectral method}\footnote{A \emph{spectral method} refers to any construction method which specifies a suitable Walsh spectrum (its Walsh support and the signs of nonzero spectral values) which after applying the inverse WHT gives a  Boolean function with prescribed properties.}.
Utilizing Proposition \ref{prop:Hrow}, the following result generalizes \cite[Theorem 3.1]{Secondary} and shows how one can efficiently construct plateaued functions (using the spectral method) by using the representation of a Walsh support as in (\ref{DPL}).
\begin{theo}\label{theo:plateH}
Let  $S_f=v+ EM=\{\omega_0,\ldots,\omega_{2^{n-s}-1}\} \subset \FB^n$, 
for some $v \in \mathbb{F}^n_2$, $M\in GL(n,\mathbb{F}_2)$ and subset $E=\{e_0,e_1, \ldots, e_{2^{n-s}-1}\}\subset \mathbb{F}^n_2$. For a function $g:\FB^{n-s}\rightarrow \mathbb{F}_2$ such that  $wt(g)=2^{n-s-1}+ 2^{\frac{n-s}{2}-1}$ or $wt(g)=2^{n-s-1}-2^{\frac{n-s}{2}-1}$, let the Walsh spectrum of $f$ be defined (by identifying $x_i \in \FB^{n-s}$ and $\omega_i \in S_f$ through $e_i \in E$) as
\begin{eqnarray}\label{eq:walshX}
W_f(u)= \left \{  \begin{array}{ll}
  2^{\frac{n+s}{2}}(-1)^{g(x_i)} & \textnormal{ for } u= v + e_iM \in S_f,\\
 0 & u \not \in S_f. \\
\end{array}  \right.
\end{eqnarray}
Then:
\begin{enumerate}[i)]
\item $f$ is an $s$-plateaued function if and only if $g$ is at bent distance to $\Phi_f$ defined by (\ref{RSf}).
\item If $E\subset \mathbb{F}^n_2$ is a linear subspace, then $f$ is an $s$-plateaued function if and only if $g$ is a bent function on $\mathbb{F}^{n-s}_2$.
\end{enumerate}
\end{theo}
\proof $i)$ By relations (\ref{WHT}) and (\ref{eq:walshX}), for arbitrary $u\in \mathbb{F}^n_2$, we have that
\begin{eqnarray}\label{eq:g1}
2^n(-1)^{f(u)}=\sum_{\omega\in S_f}W_f(\omega)(-1)^{u\cdot \omega}=2^{\frac{n+s}{2}}(-1)^{u\cdot v}\sum_{e\in E}(-1)^{f^*(v+eM)+uM^T\cdot e}
\end{eqnarray}
Since we have the correspondence  $f^*(\omega_i)=f^*(v+e_iM)=g(x_i)$ and $uM^T\cdot e_i=\phi_{uM^T}(x_i)\in \Phi_f$ ($x_i\in \mathbb{F}^{n-s}_2$, $e_i\in E$, $i\in[0,2^{n-s}-1]$), then using (\ref{eq:g1})  we get
\begin{eqnarray}\label{eq:g2}
2^{\frac{n+s}{2}}(-1)^{u\cdot v}\sum_{e\in E}(-1)^{f^*(v+eM)+uM^T\cdot e}=2^{\frac{n+s}{2}}(-1)^{u\cdot v}\sum_{x_i\in \mathbb{F}^{n-s}_2}(-1)^{g(x_i)+\phi_{uM^T}(x_i)}.
\end{eqnarray}
Thus,  we have that $f$ is $s$-plateaued if and only if $\sum_{x_i\in \mathbb{F}^{n-s}_2}(-1)^{g(x_i)+\phi_{uM^T}(x_i)}=2^{\frac{n-s}{2}} (-1)^{\tau_u}$, for some $\tau_u\in \{0,1\}$. However,  the bent distance between $g$ and $\phi_{uM^T}\in\Phi_f$ in (\ref{eq:g2}) is in accordance with relation (\ref{eq:g1}) since these two relations imply that $(-1)^{f(u)}=(-1)^{u\cdot v+ \tau_u}$ holds. Thus, the inverse WHT applied to the spectrum given by (\ref{eq:walshX}) generates a Boolean function which is $s$-plateaued.

$ii)$ If $E$ is a linear subspace (ordered lexicographically),  by Proposition \ref{prop:Hrow} the sequence profile $\Phi_f$ contains sequences of all linear functions (or their affine equivalents) and  $g$ has to be a bent function.\qed

Essentially, combining Theorem \ref{theo:plateH}  with results given in Section \ref{sec:equivalence} gives us the possibility  to construct inequivalent plateaued functions regardless of whether they are trivial or not. For instance, the construction of inequivalent trivial plateaued functions can be deduced as follows. We first remark  that for any two trivial $s$-plateaued functions $f,h:\mathbb{F}^{n}_2\rightarrow \mathbb{F}_2$ one can always relate their Walsh supports as $S_h=c+S_fA^T$, for  a matrix $A\in GL(n,\mathbb{F}_2)$ and some vector $c\in \mathbb{F}^n_2$. Then, one simply employs EA-inequivalent bent functions in $n-s$ variables as duals for $f$ and $h$ and the inequivalence follows by Theorem \ref{th:EA1}.

On the other hand, the construction of inequivalent non-trivial plateaued functions (by Theorem \ref{theo:plateH}) is achieved by setting non-affine Walsh supports of different dimensions (in combination with suitable duals).
%
The following example illustrates the former case, thus employing Theorem \ref{th:EA1} for this purpose.
\begin{ex}\label{ex:ineq1}
Let $f:\mathbb{F}^7_2\rightarrow \mathbb{F}_2$ ($n=7$) be a semi-bent function ($s=1$) defined as $f(x_1,\ldots,x_7)=x_1 x_4 + x_2 x_5 + x_3 x_6 + x_4 x_5 x_6$. One can verify that the Walsh support of $f$ is $S_f=v+E=\mathbb{F}^6_2\times \{0\}\subset \mathbb{F}^7_2$, where $v=\textbf{0}_7$ and $e_i\in E$ are given as $e_i=(x_i,0)$ with $x_i\in \mathbb{F}^6_2$ (and $\mathbb{F}^6_2$ is ordered lexicographically). In addition, the dual $f^*(\omega_i)=f^*(v+e_i)\stackrel{(\ref{DPL})}{=}\overline{f}^*(x_i)$ ($\overline{f}^*:\mathbb{F}^6_2\rightarrow \mathbb{F}_2$) is given as
$$\overline{f}^*(x_1,\ldots,x_6)=(x_1,x_2,x_3)\cdot (x_4,x_5,x_6)+x_1x_2x_3.$$

To construct a (trivial) semi-bent function $h:\mathbb{F}^5_2\rightarrow \mathbb{F}_2$ which is EA-inequivalent to $f$,  by Theorem \ref{th:EA1}, it is sufficient to take a bent function $\overline{h}^*:\mathbb{F}^6_2\rightarrow \mathbb{F}_2$ as its dual which is affine inequivalent to $\overline{f}^*$. For instance, let us define $\overline{h}^*$ as
$$\overline{h}^*(x_1,\ldots,x_6)=(x_1,x_2,x_3)\cdot (x_4,x_5,x_6).$$
Clearly, $\overline{f}^*$ and $\overline{h}^*$ are inequivalent since they are of different degrees.  Furthermore, taking that $S_h=(1,0,1,0,1,1,1)+S_f=(1,0,1,0,1,1,1)+\mathbb{F}^6_2\times \{0\}$, then by Theorem \ref{theo:plateH} (considering $h^*(z_i)=h^*((1,0,1,0,1,1,1)+(x_i,0))= \overline{h}^*(x_i)$, $x_i\in \mathbb{F}^6_2$) we obtain the function
$$h(x_1,\ldots,x_7)=x_1 + x_3 + x_1 x_4 + x_5 + x_2 x_5 + x_6 + x_3 x_6 + x_7.$$
Clearly, $f$ and $h$ are inequivalent due to $deg(f)\neq deg(h)$.
\end{ex}
Some generic construction methods,  using the spectral approach of Theorem \ref{theo:plateH}, of nontrivial plateaued functions are presented in Section \ref{sec:nonaffine}. Two EA-inequivalent nontrivial plateaued functions can be constructed by ensuring that their Walsh supports contain different number of linearly independent vectors (thus applying Theorem \ref{th:eq}-$(i)$). Alternatively, if their supports have the same number of linearly independent vectors then the idea from Example \ref{ex:ineq1} can be used where additionally the condition of Theorem \ref{th:eq}-$(ii)$ on duals has to be satisfied.

\subsection{Plateaued functions without linear structures}\label{sec:linear}

It has already been observed that trivial plateaued functions essentially correspond to partially bent functions {\cite{Zheng,Ayca2}} introduced by Carlet in \cite{CarletPartial}. Furthermore, it is well-known that partially bent functions  always admit linear structures. Therefore, it is of interest to analyze whether  nontrivial plateaued functions admit linear structures which extends the work of Zhang and Zheng \cite{Zheng} where a design method of nontrivial plateaued functions without linear structures were given in terms of sequences.
%
A useful concept related to a Boolean function $f \in \FB^n$ is the so-called autocorrelation spectrum whose spectral values are given by  $\triangle_f(a)=\sum_{x\in \FB^n}(-1)^{f(x)+f(x+a)}$, for $a \in {\FB^n}^*$. The property of being partially bent can be stated as follows.

\begin{defi}\label{def partially bent}\cite{CarletPartial}
  A function $f\in B_n$ with the Walsh support $S_f$  which satisfies the equality
  $\# S_f\cdot \# \mathfrak{R}=2^n$, where $\mathfrak{R}=\{u\in \mathbb{F}^n_2:\; \Delta_f(u)\neq 0\}$,
  is called partially-bent.
\end{defi}
In has been shown  \cite[Corollary 3.1]{CCA}, using somewhat different notation, that the autocorrelation values can be computed in terms of spectral values over the Walsh support.
\begin{lemma}\cite{CCA}\label{lemma:connection}
Let $f \in \mathcal{B}_n$ be an arbitrary functions with Walsh support $S_f$ represented as $S_f=v+ E$, for some subset $E$ of $\FB^n$ and $v\in S_f$. Then
\begin{equation}\label{eq:connect}
\triangle_f(a)=2^{-n}(-1)^{v\cdot a}\sum\limits_{u\in E}W_f^2(u+v)(-1)^{u\cdot a}.
\end{equation}
\end{lemma}
The following result specifies the existence of linear structures for an arbitrary Boolean function, thus including plateaued functions as a special case.
 \begin{theo}\label{theo baselinear}
Let $f \in \mathcal{B}_n$  whose Walsh spectrum
is given by  $S_f= v+E\subset \FB^n$, where  $v\in S_f$.
Let $m$ denote the maximum number  of linearly independent elements of $E$ and  $\Lambda=\{a\in \FB^n: f(x)+f(x+a)=constant, x\in \FB^n \}$. Then $$m+\dim(\Lambda)=n.$$
\end{theo}
\begin{proof}
  If $a $ is a linear structure of $f$,
  then $\triangle_f(a)=2^n ~\emph{or} ~-2^n$.  By relation (\ref{eq:connect}), this is equivalent to \begin{equation}\label{equ lsup00}
\sum\limits_{u\in \FB^n}W_f^2(u)(-1)^{u\cdot a}= (-1)^{v \cdot a}\sum\limits_{u\in E}W_f^2(u+v)(-1)^{u\cdot a}=2^{2n} ~\emph{or} ~-2^{2n}.\end{equation}
Using the Parseval's equality $ \sum\limits_{u\in \FB^n}W_f^2(u)=2^{2n}$ in (\ref{equ lsup00}), we have that
$$\sum\limits_{u\in E}(-1)^{u\cdot a}=\# S_f\; \textnormal{ or }\;-\# S_f. $$
Consequently, using the fact that $S_f=v+E$ and $v\in S_f$, we have that $0_n\in E$ and thus the sum on the left-hand side above contains the term $(-1)^{0_n \cdot a}=1$.
%
Hence, $$ \sum\limits_{u\in E}(-1)^{u\cdot a}=\# S_f.$$

 On the other hand,  $\sum\limits_{u\in E}(-1)^{u\cdot a}=\# S_f$ if and only if $E\subseteq \{a\}^{\perp}$.  Similarly,  we have $E\subseteq \Lambda^{\perp}$, that is,
  \begin{equation}\label{equ lsup 1}\dim(E_*)\leq n-\dim(\Lambda),
   \end{equation}where $E_*$ denotes the minimum subspace including $E$.
From (\ref{eq:connect}), we know that  $|\triangle_f(a)|=2^n$ for any $a\in E_*^{\perp} $. Thus, we have $ E_*^{\perp}\subseteq\Lambda$, that is,
\begin{equation}\label{equ lsup 2} \dim(\Lambda)\geq n-\dim(E_*).
   \end{equation}
  Combining (\ref{equ lsup 1}) and (\ref{equ lsup 2}),  we have $m+\dim(\Lambda)=n.$ \qed

\end{proof}

 \begin{cor}\label{theo:baselinear2}
Let $f:\mathbb{F}^n_2\rightarrow \mathbb{F}_2$ be an arbitrary $s$-plateaued function with its Walsh spectrum given as  $S_f= v+E\subset \FB^n$, where  $v\in S_f$.
Then $f$ does not contain (non-zero) linear structures if and only if the number of linearly independent vectors of $E$ is equal to $n.$
\end{cor}
\begin{proof}
Let us assume that $f$ has a non-zero linear structure $a\in \mathbb{F}^n_2$. Then by relation (\ref{eq:connect}) we have that
$$\Delta_f(a)=2^{-n}(-1)^{v\cdot a}2^{n+s}\sum_{u\in E}(-1)^{u\cdot a}=2^{s}(-1)^{v\cdot a}\sum_{u\in E}(-1)^{u\cdot a},$$
and thus $a$ is a linear structure of $f$ (that is $|\Delta_f(a)|=2^n$) if and only if $\sum_{u\in E}(-1)^{u\cdot a}=\pm 2^{n-s}$. However, using the fact that $\# E=2^{n-s}$, we have that
\begin{eqnarray}\label{eq:phi}
\pm 2^{n-s}=\sum_{u\in E}(-1)^{u\cdot a}=2^{n-s}-2wt(\varphi_a),
\end{eqnarray}
where the truth table of the function $\varphi_a$ is given as $T_{\varphi_a}=(u\cdot e_0,\ldots,u\cdot e_{2^{n-s}-1})$ (and $E=\{e_0,\ldots,e_{2^{n-s}-1}\}$ can be considered in any ordering). However, (\ref{eq:phi}) holds if and only if $\varphi_a$, viewed as a function on $\mathbb{F}^{n-s}_2$, is a constant function for $a\neq \textbf{0}_n$ (thus equal to $0$ or $1$).

 Equivalently, representing $E$ as a matrix, using $e_0,\ldots,e_{2^{n-s}-1}$ as its rows, then some of its columns are either constant or there is a linear dependency among its columns. Thus, the number $E$ of linearly independent columns of $E$ is less than $n$.
\qed
\end{proof}
It is rather challenging to conjecture that only trivial plateaued functions admit linear structures which is however not true as illustrated in the example below.
\begin{ex}
Let us consider the function $f:\mathbb{F}^6_2\rightarrow \mathbb{F}_2$ $(n=6)$ given as $$f(x_1,\ldots,x_6)=x_1 (x_3 + x_2 x_5) + x_2 (x_4 + x_6).$$
One can verify that the autocorrelation spectrum of $f$ is given by
$$(\Delta_f(u_0),\ldots,\Delta_f(u_{63}))=(64, 0, 32, 32, 0, 64, 32, 32, 0, 0, 32, -32, 0, 0, -32, 32,\bf{0}_{48}),$$
where $\bf{0}_{48}$ means that the last $48$ coefficients are equal to $0$. Also, the Walsh support of $f$ is given as $S_f=\mathbb{F}^4_2\wr T_g\wr T_{\ell}$, where $g,\ell:\mathbb{F}^4_2\rightarrow \mathbb{F}_2$ are given as $g(x_1,\ldots,x_4)=x_3x_4$ and $\ell(x_1,\ldots,x_4)=x_4$, and thus $f$ is a non-trivial plateaued function. Since $S_f$ contains only $5$ linearly independent vectors, then $dim(\Lambda)=n-m=6-5=1$ (Theorem \ref{theo baselinear}) and we have that $\Lambda=\{{\bf 0}_6,(0, 0, 0, 1, 0, 1)\}.$
\end{ex}
Clearly, if the Walsh support of $f$ contains $n$ linearly independent vectors then $\dim(\Lambda)=0$ and such a function cannot have linear structures. This result will be useful in Section \ref{sec:nontrivial} to draw a conclusion  that the constructed nontrivial plateaued functions are without linear structures.
\begin{rem}
The above result implies that a trivial $s$-plateaued function on $\FB^n$ (which corresponds to a partially bent function) always admit linear structures of dimension $s$, see also \cite{Zheng}. This is a consequence of the fact that $\dim(S_f)=n-s=m$ and thus $\dim(\Lambda)=s$.
\end{rem}

\section{Constructions of (non)trivial disjoint spectra plateaued functions} \label{sec:nontrivial}

A construction method of  trivial disjoint spectra plateaued functions is easily established by using Theorem \ref{theo:plateH} and setting the Walsh supports to be disjoint affine subspaces (where duals can be arbitrary bent functions).
The design of nontrivial plateaued functions and in particular the problem of deriving  disjoint spectra nontrivial $s$-plateaued functions of maximum cardinality appears to be a more difficult task. In this direction, we provide some efficient design methods  of  nontrivial plateaued functions  which also might have  interesting implications on the design of bent functions.

\subsection{Disjoint spectra plateaued functions  via $\mathcal{MM}$ class}\label{sec:MM}

For convenience and due to the consequences related to generic design methods of bent functions, we briefly discuss the possibility of obtaining trivial disjoint spectra plateaued functions of maximal cardinality. Since any trivial $s$-plateaued function has an affine subspace $S_f=\alpha + E$ of dimension $n-s$ as its Walsh support, it is natural to consider the space $\FB^n$ as a union of disjoint cosets, thus $\FB^n=\cup_{\alpha_i \in \mathcal{G}} \alpha_i +E$ ($\mathcal{G}$ is orthogonal complement of $E$), for some linear subspace $E$ and elements $\alpha_i$ such that $\alpha_i + E \cap \alpha_j +E = \emptyset$. Then, clearly $\#\mathcal{G}=2^{s}$ and any function $f_i \in B_n$ whose support is defined on the affine subspace $S_{f_i}=\alpha_i +E$ is a disjoint trivial plateaued function to $f_j$, for $\alpha_i \neq \alpha_j$. It is important to notice that such a set of disjoint spectra plateaued functions (of maximal cardinality) can be defined using  arbitrary  bent duals for any $f_i$ which additionally may belong  to different classes of bent functions. This leads to an important research problem that can be stated as follows.
\begin{op}
Assume that a  set of disjoint spectra trivial $s$-plateaued functions on $\FB^n$ of (maximal) cardinality $2^s$ is constructed using the dual bent functions (by means of Theorem \ref{theo:plateH}) which are taken from different classes of bent functions on $\FB^{n-s}$.  To which class of bent functions a bent function on $\FB^{n+s}$, defined as a concatenation of these $2^s$ functions, do belong?
\end{op}
 The above approach, specifying the bent dual in spectral domain, provides much more flexibility  compared to a standard design in the ANF domain. The only generic class of Boolean functions that admits an efficient construction method of disjoint spectra plateaued functions of maximal cardinality in the ANF domain is the Marioana-McFarland ($\mathcal{MM}$) class \cite{MMclass}. This calls  was initially introduced as a primary class of bent functions though it can be used for plateaued functions as well.

Let $n=s+k$ where $s > n/2$ and denote the set of linear functions in $s$ variables as $\mathcal{L}_s=\{a_1x_1 + \ldots + a_s x_s: a_i \in \mathbb{F}_2\}$. Let $f(y,x)= \phi(y) \cdot x + g(y)$, where $x=(x_1,\ldots,x_s) \in \FB^s$, $y=(y_1,\ldots,y_k) \in \FB^k$, and $\phi : \FB^k \rightarrow \FB^s$ be an injective mapping. Since $s>k$ we can partition the vector space $\FB^s$ of cardinality $2^s$ into $2^{s-k}$ disjoint subsets $E_i$ each of cardinality $2^k$. That is, $\FB^s=\cup_i^{2^{s-k}}E_i$, where $E_i \cap E_j = \emptyset$.
Then we can define  a set of functions
\begin{equation} \label{eq:mmdisjoint}
f^{(i)}(y,x)=\phi_i(y) \cdot x + g_i(y), \;\;  i=1, \ldots, 2^{s-k},
\end{equation}
where  $\phi_i : \FB^k \rightarrow E_i \subset \FB^s$ is a bijection and $g_i$ is arbitrary. It is straightforward to verify that each $f_i$ is $s$-plateaued (see also \cite{CarletQ}) and that $f_i$ are disjoint spectra functions.


\begin{prop}\label{prop:disspec_affine}
The set of functions $f^{(i)}$, for $i=1, \ldots, 2^{s-k}$, defined by (\ref{eq:mmdisjoint}) is a set of disjoint spectra plateaued functions. Furthermore, if the sets $E_i$ used to define the functions $\phi_i : \FB^k \rightarrow E_i \subset \FB^s$ are not affine subspaces then $f_i$ are nontrivial disjoint spectra plateaued functions.
\end{prop}

\begin{rem}
It can be easily checked that if $E_i$ are  affine subspaces then  the concatenation of $f^{(1)}, \ldots,f^{(2^{s-k})}$  will again provide a bent function in the $\mathcal{MM}$ class.
\end{rem}

\subsection{Construction of nontrivial plateaued functions}\label{sec:nonaffine}

An example of a plateaued function with non-affine Walsh support (which is not partially bent and thus nontrivial), has been given in \cite[Example 3.2]{Secondary}. In difference to the approach taken in the previous section, using again  the $\mathcal{MM}$ class of bent functions, we will  provide a generic construction of nontrivial plateaued functions in the spectral domain. Furthermore, we show that the classes $\mathcal{C}$ and $\mathcal{D}$ \cite{CarletTNC} can also be used for the same purpose, but with certain limitations (Remark \ref{rem:CD}).
Before we present a general framework, we illustrate the main construction idea of nontrivial plateaued functions through an example. 
In what follows,  the truth table $T_g$ of a function $g\in \mathcal{B}_n$ will be viewed as a column vector of length $2^n$.
\begin{ex}\label{ex1}
Using Theorem \ref{theo:plateH}, one can construct a plateaued function with non-affine Walsh support as follows. Let us define the  columns of the Walsh support $S_f$ as
$$S_f=T_{\phi_{b_1}}\wr\ldots\wr T_{\phi_{b_4}}\wr T_{\phi_{b_5}}=\mathbb{F}^4_2\wr T_{\phi_{b_5}},$$
where $\phi_{b_5}:\mathbb{F}^4_2\rightarrow \mathbb{F}_2$ is defined as $\phi_{b_5}(x_1,\ldots,x_4)=x_3x_4$.  Let also the dual $\overline{f}^*=g$ be defined as $g(x_1,\ldots,x_4)=x_1x_3+ x_2x_4\in \mathcal{MM}$, see Table \ref{tab:ex1}. Notice that the vectors $x=(x_1,\ldots,x_4)$ are sorted lexicographically.
\begin{table}[H]
\scriptsize
\centering
\begin{tabular}{|r|c|c|}
  \hline
  \makecell{$S_f=\mathbb{F}^4_2\wr T_{\phi_{b_5}}$} & $f^*(\omega_i)=\overline{f}^*(x_i)=g(x_i),$ $x_i\in \mathbb{F}^4_2$\\ \hline
  $\omega_0=(0, 0, 0, 0, 0)$  & 0 \\
  $\omega_1=(0, 0, 0, 1, 0)$ &  0 \\
  $\omega_2=(0, 0, 1, 0, 0)$  & 0 \\
  $\omega_3=(0, 0, 1, 1, 1)$  & 0 \\
  $\omega_4=(0, 1, 0, 0, 0)$  & 0 \\
  $\omega_5=(0, 1, 0, 1, 0)$  & 1 \\
  $\omega_6=(0, 1, 1, 0, 0)$  & 0 \\
  $\omega_7=(0, 1, 1, 1, 1)$  & 1 \\
  $\omega_8=(1, 0, 0, 0, 0)$  & 0 \\
  $\omega_9=(1, 0, 0, 1, 0)$  & 0 \\
$\omega_{10}=(1, 0, 1, 0, 0)$  & 1 \\
$\omega_{11}=(1, 0, 1, 1, 1)$  & 1 \\
$\omega_{12}=(1, 1, 0, 0, 0)$  & 0 \\
$\omega_{13}=(1, 1, 0, 1, 0)$  & 1 \\
$\omega_{14}=(1, 1, 1, 0, 0)$  & 1 \\
$\omega_{15}=(1, 1, 1, 1, 1)$  & 0 \\
  \hline
\end{tabular}
\caption{Values of $f^*$ with respect to $S_f$.}
\label{tab:ex1}
\end{table}
Then, by applying the inverse WHT where  the non-zero Walsh spectral values are given with respect to Table \ref{tab:ex1} and using (\ref{eq:walshX}), we obtain a semi-bent function $f$ given as
$$f(x_1,\ldots,x_5)=x_1 x_3 + x_2 x_4 + x_1 x_2 x_5.$$
It can be verified that $\#\mathfrak{R}=5$, hence $f$ is not a partially bent function.
\end{ex}
\begin{rem} The function $f$ in Example \ref{ex1} does not admit linear structures. This is the consequence of Theorem \ref{theo baselinear} and the fact that $S_f$ contains five linearly independent vectors, thus spanning $\FB^5$.
\end{rem}
In general, the construction of $S_f$ by means of Theorem \ref{theo:plateH} has to satisfy two important conditions. More precisely, $S_f$ must not  be a multiset and its  columns (when $S_f$ is represented as a matrix of size $2^{n-s} \times n$) must lie at bent distance to the dual $\overline{f}^*=g$. The  requirement that $S_f=\mathbb{F}^4_2\wr T_{\phi_{b_5}}$ in Example \ref{ex1} is not a multi-set is apparently fulfilled regardless of the choice of  $\phi_{b_5}$. The second condition is satisfied since the first four columns of $S_f$ are sequences or linear functions (\cite[Lemma 3.1]{Secondary}). Furthermore, $g+\phi_{b_5}$ is a bent function in $n-s$ variables which belongs to the $\mathcal{MM}$ class and therefore $g$ and $\phi_{b_5}$ are at bent distance.

In order to construct a plateaued function $f$ on $\FB^n$ (of any amplitude) we will employ the $\mathcal{MM}$ class of bent functions on $\FB^{n-s}$ and define   its dual as $\overline{f}^*=g\in \mathcal{MM}$. This approach will be later extended  to cover the cases  when $\overline{f}^*=g$ belongs to $\mathcal{C}$ or $\mathcal{D}$ \cite{CarletTNC} classes of bent functions. The main idea behind  our approach is to specify  the Walsh support of $f$ as  $S_f=\mathbb{F}^{n-s}_2\wr S$, for some suitable set $S$ of cardinality $2^{n-s}$ whose elements belong to $\FB^{s}$.

To construct  nontrivial plateaued functions one can employ  bent functions from the  $\mathcal{MM}$ class defined by \begin{equation} \label{def:bentmm}
g(x,y)=x\cdot \psi(y)+ t(y); \;\;x,y\in \mathbb{F}^{(n-s)/2}_2,
\end{equation}
where $\psi$ is an arbitrary permutation on $\mathbb{F}^{(n-s)/2}_2$ and $t \in \mathcal{B}_{(n-s)/2}$ is arbitrary.
\begin{theo}\label{th:Sf}
Let $g:\mathbb{F}^{n-s}_2\rightarrow \mathbb{F}_2$ be defined by (\ref{def:bentmm}).
For an arbitrary matrix $M\in GL(n-s,\mathbb{F}_2)$ and vector $c\in \mathbb{F}^{n-s}_2$, let $S_f$ be defined as
\begin{eqnarray}\label{Sff}
S_f=T_{\phi_{b_1}}\wr\ldots \wr T_{\phi_{b_{n-s}}} \wr T_{\phi_{b_{n-s+1}}}\wr \ldots\wr T_{\phi_{b_n}}=(c+EM)\wr T_{t_1}\wr\ldots\wr T_{t_s},
\end{eqnarray}
where $t_i=t_i(y)\in \mathcal{B}_{(n-s)/2}$ $(i\in[1,s])$ are arbitrary and $E=\mathbb{F}^{n-s}_2$. Then
$f:\mathbb{F}^n_2\rightarrow \mathbb{F}_2$ constructed using Theorem \ref{theo:plateH}, with $S_f$ defined by (\ref{Sff}) and taking $\overline{f}^*=g$,  is an $s$-plateaued function.
\end{theo}
\proof  By Proposition \ref{prop:Hrow}, we have that the columns corresponding to $c+EM$  are  affine functions in $n-s$ variables. Thus, we may write $\phi_{b_i}=\ell_i$  for some $\ell_i\in \mathcal{A}_{n-s}$, where $i\in[1,n-s]$.
 Since $g$ is a bent function in the $\mathcal{MM}$ class, then for any $v\in \mathbb{F}^k_2$
 $$g(x,y)+ v\cdot (\phi_{b_1}(x,y),\ldots,\phi_{b_n}(x,y))=g(x,y)+ v\cdot(\ell_1(x,y),\ldots,\ell_{n-s}(x,y),t_1(y),\ldots,t_s(y))$$
is also a bent function. By Theorem \ref{theo:plateH}-$(i)$, this  means that $g=\overline{f}^*$ is at bent distance to functions $v\cdot (\phi_{b_1},\ldots,\phi_{b_n})$, for every $v\in \mathbb{F}^n_2$. Thus, $f$ is $s$-plateaued.\qed

Using the same arguments as in the proof of Theorem \ref{th:Sf}, one can easily employ $\mathcal{C}$ and $\mathcal{D}$ classes of bent functions in order to construct semi-bent functions ($s=1$) as follows.
\begin{theo}\label{th:CD}
Let $g(x,y)=x\cdot \psi(y)$, $x,y\in \mathbb{F}^{n/2}_2$, be a bent function, $n$ is even.
For an arbitrary matrix $M\in GL(n,\mathbb{F}_2)$ and vector $c\in \mathbb{F}^{n}_2$, let $S_f=(c+EM)\wr T_{\mu},$ where $E=\mathbb{F}^{n}_2$ is ordered lexicographically and $\mu\in \mathcal{B}_{n}$. Then:
 \begin{enumerate}[i)]
\item  If for subspaces $E_1,E_2\subset \mathbb{F}^{n/2}_2$ we have $\mu(x,y)=\phi_{E_1}(x)\phi_{E_2}(y)$ and $\psi$ satisfies the equality $\psi(E_2)=E^{\perp}_1$, then the function $f:\mathbb{F}^n_2\rightarrow \mathbb{F}_2$ with Walsh support $S_f$ and the dual $\overline{f}^*=g$ is a semi-bent function.
\item For a subspace $L\subset \mathbb{F}^{n}_2$ let $\mu(x,y)=\phi_L(x)$. If $\psi^{-1}(v+L^{\perp})$ is an affine subspace for all $v\in \mathbb{F}^{n}_2$, then the function $f:\mathbb{F}^n_2\rightarrow \mathbb{F}_2$ with Walsh support $S_f$ and the dual $\overline{f}^*=g$ is a semi-bent function.
\end{enumerate}
\end{theo}
\begin{rem}\label{rem:CD}
Note that Theorem \ref{th:CD} can be modified so that it provides $s$-plateaued functions with $s>1$, similarly as Theorem \ref{th:Sf}. However, in that case (setting $g$ as in Theorem \ref{th:CD}) in relation (\ref{Sff}) we need to set functions $t_i$ to be equal to $t_i(x,y)=\phi_{E_1}(x)\phi_{E_2}(y)$ or $t_i(x,y)=\phi_L(x)$ for all $i\in[1,s]$. The other possibility is to find a space of functions $t_i$ which have these forms and for which the permutation $\psi$ satisfies the necessary conditions imposed by the definitions of $\mathcal{D}$ and $\mathcal{C}$ classes (since the dual $\overline{f}^*=g$ has to be at bent distance to $\Phi_f$).
\end{rem}
The above approach  can be easily generalized using  the following result.
\begin{theo}\label{theo:vectorial}
Let $H=(h_1,\ldots,h_{\lambda}):\mathbb{F}^{n-s}_2\rightarrow \mathbb{F}^{\lambda}_2$, where  $\lambda\leq (n-s)/2$ and $n-s$ is even, be a vectorial bent function. Let the dual of $f$  be specified as $\overline{f}^*=g=h_i$, for some fixed $i\in[1,\lambda]$. In addition, for an arbitrary matrix $M\in GL(n-s,\mathbb{F}_2)$ and vector $c\in \mathbb{F}^{n-s}_2$, let $S_f$ be defined as
$$S_f=(c+EM)\wr T_{t_1}\wr\ldots\wr T_{t_m} \wr T_{h_1}\wr\ldots\wr T_{h_r},$$
where $t_i\in \mathcal{A}_{n-s}$, $m+r=s$ with $r\leq \lambda-1$ and $E=\mathbb{F}^{n-s}_2$ is ordered lexicographically. In addition, $h_1,\ldots,h_{i-1},h_{i+1},\ldots, h_r \neq g$.
Then,  $f:\mathbb{F}^n_2\rightarrow \mathbb{F}_2$, whose  Walsh support is $S_f$ and the dual $\overline{f}^*=g=h_i$,
is an $s$-plateaued function.
\end{theo}
By Proposition \ref{prop:Hrow}, if the affine subspace $c+EM$ (with $E=\mathbb{F}^{n-s}_2$)
is written as a matrix of size $2^{n-s}\times (n-s)$, then its columns  correspond to linear/affine functions on $\mathbb{F}^{n-s}_2$. An important  question in this context is whether there exists (many) plateaued functions with the property that $S_f=T_{\phi_{b_1}}\wr \ldots \wr T_{\phi_{b_n}}$ ($\phi_{b_i}$ defined by (\ref{RSf})) contains strictly less than $n-s$ columns which correspond to affine functions on $\mathbb{F}^{n-s}_2$.
\begin{op}
Provide (generic) constructions of $s$-plateaued functions in $n$ variables whose Walsh support, when written as a matrix of order $2^{n-s}\times n$, contains more than $s$ columns which correspond to nonlinear  functions on $\mathbb{F}^{n-s}_2.$
\end{op}
\begin{rem}
Recall that  $S_f\subset \mathbb{F}^n_2$ of cardinality $2^{n-s}$ can be used as a Walsh support for an $s$-plateaued function if and only if
there exists a function  $f^*\in \mathcal{B}_{n-s}$ of weight $2^{n-s-1}\pm 2^{(n-s)/2-1}$ which is at bent distance to all  linear combinations of columns of $S_f$ (Theorem \ref{theo:plateH}).
\end{rem}

\subsection{Designing disjoint spectra nontrivial plateaued functions}

The problem of finding  nontrivial  $s$-plateaued functions with disjoint spectra is intrinsically more difficult than for trivial plateaued functions.
 The reason is   that taking $S_f$ not to be affine subspace than the cosets of $S_f$ are in general not disjoint, that is, $(v+ S_f)\cap S_f \neq \emptyset$ for some $v\in\mathbb{F}^n_2$.
Nevertheless, the class of nontrivial plateaued functions given in Theorems \ref{th:Sf}-\ref{theo:vectorial} also gives rise to disjoint spectra nontrivial plateaued functions.
\begin{lemma}\label{lemma:Q}
Let $f$ be an $s$-plateaued functions constructed by any of Theorems \ref{th:Sf}-\ref{theo:vectorial} with $S_f \subset \FB^n$, where $\#S_f=2^{n-s}$. Define the set $Q$ of cardinality $2^s$ by
\begin{eqnarray}\label{Q}
Q=\{c\cdot (b_{n-s+1},b_{n-s+2},\ldots,b_{n}):c\in \mathbb{F}^s_2\},
\end{eqnarray}
where $b_i=(0,\ldots,0,1,0,\ldots,0)\in \mathbb{F}^n_2$ ($1$  at the $i$-th position).  Then, the sets $q+S_f$, $q\in Q$, partition the space $\mathbb{F}^n_2$, i.e., $$\mathbb{F}^n_2=\bigcup_{q\in Q} (q + S_f), \;\;  (q_i+ S_f)\cap(q_j+ S_f)=\emptyset \;\; \textnormal{ for } q_i\neq q_j,  \;\; q_i,q_j\in Q.$$
\end{lemma}
\proof In order to prove the statement for all mentioned theorems, we assume that $S_f$ is given as in (\ref{Sff}) and additionally  we take that $t_i=t_i(x,y)$ (thus depending also on $x$). Then, such a Walsh support $S_f$ covers all the cases in Theorem \ref{th:Sf}-\ref{theo:vectorial}.

It is sufficient to prove that $(q_i+ S_f)\cap(q_j+ S_f)=\emptyset$, for $q_i\neq q_j$ ($q_i,q_j\in Q$), since $\# Q=2^s$ and $\# S_f=2^{n-s}$ will then imply that $\mathbb{F}^n_2=\bigcup_{q\in Q} (q+ S_f)$.  Assume, on contrary,  that for some $q_i\neq q_j$ there exists  $w\in (q_i+ S_f)\cap(q_j+ S_f)$, where $S_f$ is given by relation (\ref{Sff}).
Then, denoting  $q_i=(q^{(i)}_1,\ldots,q^{(i)}_s)$ and $q_j=(q^{(j)}_1,\ldots,q^{(j)}_s)$, for some $x',x'',y',y''\in \mathbb{F}^{(n-s)/2}_2$ and $M\in GL(n,\mathbb{F}^{n-s}_2)$ it holds that
\begin{eqnarray*}
\left\{\begin{array}{c}
         w=(c+(x',y')M, t_1(x',y')+ q^{(i)}_1,\ldots,t_s(x',y')+ q^{(i)}_s)\in q_i+ S_f \\
         w=(c+(x'',y'')M, t_1(x'',y'')+ q^{(j)}_1,\ldots,t_s(x'',y'')+ q^{(j)}_s)\in q_j+ S_f
       \end{array}
\right..
\end{eqnarray*}
This implies that $x'=x''$, $y'=y''$ and $t_r(x',y')+ q^{(i)}_r=t_r(x'',y'')+ q^{(j)}_r$, for all $r\in[1,s]$, and consequently $q_i=q_j$. This contradicts the assumption that $q_i \neq q_j$ (since $t_r(x',y')=t_r(x'',y'')$), which completes the proof.\qed
Now employing Theorems \ref{th:Sf}-\ref{th:CD} and Lemma \ref{lemma:Q} we derive a construction method of disjoint spectra nontrivial plateaued functions (possibly without linear structures). Additionally, we show that these functions can be efficiently utilized for construction of bent functions.
%
%
\begin{theo}\label{th:gencon}
Let $V$ be a $n$-dimensional subspace of $\mathbb{F}^{m}_2$ ($m$ even), such that $W\oplus V=\mathbb{F}^m_2$,
where $W=\{a_0,a_1,\ldots,a_{2^{m-n}-1}\}$ and $Q=\{q_0,q_1,\ldots,q_{2^s-1}\}$ is defined by (\ref{Q}) ($s=m-n$, $q_0=\textbf{0}_n,$ $a_0=\textbf{0}_k$). Let $f:\mathbb{F}^n_2\rightarrow \mathbb{F}_2$ be  $s$-plateaued  constructed by Theorem \ref{th:Sf} or \ref{th:CD}. Then:
\begin{enumerate}[i)]
\item Let  $f_0,\ldots,f_{2^{s}-1}:\mathbb{F}^n_2\rightarrow \mathbb{F}_2$ ($f_0=f$) be $s$-plateaued functions constructed by Theorem \ref{theo:plateH}, with Walsh supports $S_{f_i}=S_f+ q_i\subset \mathbb{F}^n_2$ and bent duals $\overline{f}^*_i$ taken from the $\mathcal{MM}$ or $\mathcal{C}$/$\mathcal{D}$ class  if Theorem \ref{th:Sf} or Theorem \ref{th:CD} is used, respectively. Then $f_0,\ldots,f_{2^{s}-1}$ are pairwise disjoint spectra functions.
\item $f_1,\ldots,f_{2^s-1}$ do not admit linear structures if and only if $f$ does not admit linear structures, that is when $dim(S_{f})=n$.
\item Suppose that the restrictions of a function $\mathfrak{f}:\mathbb{F}^m_2\rightarrow \mathbb{F}_2$ are defined by
 $$\mathfrak{f}|_{a_i+ V}(a_i+x)=f_i(x),\;\;\;i\in[0,2^{m-n}-1],\; x\in V,$$
 where $f_i$ are constructed as in $i)$. Then, $\mathfrak{f}$ is a bent function on $\mathbb{F}^m_2$.
 \end{enumerate}
\end{theo}
\proof $i)$ By Lemma \ref{lemma:Q}, for any $i\neq j$ ($i,j\in [0,2^s-1]$) it holds that $S_{f_i}\cap S_{f_j}=\emptyset$, and thus $f_0,\ldots,f_{2^s-1}$ are disjoint spectra functions.

$ii)$ By Corollary \ref{theo:baselinear2},  $f_i$ does not admit linear structures if and only if $\dim(S_{f_i})=\dim(S_f)=n$.

$iii)$ For an arbitrary $u\in \mathbb{F}^n_2$,  the WHT of $\mathfrak{f}$ is given by
\begin{eqnarray*}\label{maineq}\nonumber
W_\mathfrak{f}(u)=\sum_{a_i\in W}(-1)^{u\cdot a_i}\sum_{x\in V}(-1)^{f_i(x)+ u\cdot x}=\sum_{a_i\in W}(-1)^{u\cdot a_i}W_{f_i}(\vartheta_u),
\end{eqnarray*}
where $\vartheta_u\in \mathbb{F}^{n}_2$ is a vector for which it holds that $(\vartheta_u\cdot y_0,\ldots,\vartheta_u\cdot y_{2^n-1})=(u\cdot x_0,\ldots,u\cdot x_{2^n-1})$, where $y_i\in \mathbb{F}^n_2$, $x_i\in V\subset\mathbb{F}^m_2$. Recall that by Proposition \ref{prop:Hrow} for lexicographically ordered space $V=\{x_0,\ldots,x_{2^n-1}\}$ it holds that $(u\cdot x_0,\ldots,u\cdot x_{2^n-1})$ is a sequence of a linear function.
%
Consequently, using the fact that $f_i$ are disjoint spectra $s$-plateaued functions on $\mathbb{F}^n_2$, we have that $W_{\mathfrak{f}}(u)=\pm 2^{\frac{n+s}{2}}=\pm 2^{\frac{m}{2}}$, which completes the proof. \qed

The above result provides an efficient method of specifying nontrivial disjoint spectra $s$-plateaued functions of maximal cardinality. 
 In general, the problem of finding {\em extendable sets} of nontrivial disjoint spectra plateaued functions  seems to be quite interesting.
\begin{op}
Assume that we have a set $\{f_1,\ldots,f_r\}$ of nontrivial disjoint spectra $s$-plateaued functions defined on $\mathbb{F}^n_2$, where $r<2^s$. Is it always possible to extend this set with some $s$-plateaued functions $f_{r+1},\ldots,f_{2^s}\in \mathcal{B}_n$ so that $\{f_1,\ldots,f_r,f_{r+1},\ldots,f_{2^s}\}$ is a set of pairwise disjoint spectra $s$-plateaued functions?
\end{op}

\section{Applying nonlinear transforms to plateaued functions}\label{sec:nonlperm}

In difference to standard approaches, two decades ago Hou and Langevin \cite{HouLang} proposed quite a different framework for specifying the bent properties of Boolean functions. More precisely, considering an arbitrary Boolean functions $f$  one may ask a question what kind of nonlinear permutation $\sigma$ over $\FB^n$ composed to the input variables of $f$ yield a bent function $f \circ \sigma^{-1}$. Quite recently, their approach has been extended in \cite{Houextension2017} where the authors have provided several different methods that result in affine inequivalent bent functions. This method can also be applied to plateaued functions, provided that one can identify some classes of plateaued functions that can be represented in a suitable form. In what follows we recall this method applied to bent functions in detail.

Throughout this section a permutation $\sigma: \FB^n \rightarrow \FB^n$ is represented as a collection of its $n$ coordinate functions so that $\sigma(x)=(\sigma_1(x), \ldots, \sigma_n(x))$, where $\sigma_i:\FB^n \rightarrow \FB$.
For even $n$, denoting by $l(f) =\{g \in \mathcal{B}_n : d_H(f,g)=2^{n-1} \pm 2^{n/2-1}\}$ the set of Boolean functions at the bent distance from $f$ Hou and Langevin showed that  $f \circ \sigma^{-1}$ is bent if and only if the linear span $\langle \sigma_1,\ldots,\sigma_n \rangle =span(\sigma)$ is a subset of $l(f)$, cf. Lemma 3.1 in \cite{HouLang}. Furthermore, they also analyzed a special form of bent functions given as
\begin{equation}
\label{eq:form}
f=x_1f_1 + x_2f_2 +x_1x_2 \alpha + g,
\end{equation}
 where the functions $f_1,f_2,g$ and $\alpha$ (where  $\alpha \in \mathcal{A}_n$) only depend on variables $x_3, \ldots, x_n$.

 The main question is whether similar results can be applied to plateaued functions whose spectra is 3-valued.
\begin{theo}\label{th:sigma}
Let $f \in  \mathcal{B}_n$ be  $s$-plateaued and let $\sigma=(\sigma_1 ,\ldots, \sigma_n): \FB^n \rightarrow    \FB^n$ be a bijection. Then $f \circ \sigma^{-1}$  is an $s$-plateaued   function if and only if  $span(\sigma_1 , \ldots, \sigma_n) \subset lp( f )$, where $lp( f )=\{g \in \mathcal{B}_n: W_{f+g}(u) \in \{0,\pm 2^{\frac{n+s}{2}} \},\;\forall u\in \mathbb{F}^n_2\}$.
\end{theo}
\begin{proof}
For an arbitrary $u\in \mathbb{F}^n_2$ we have that
$$W_{f \circ \sigma^{-1}}(u) = \sum_{x \in \FB^n }(-1)^{f \circ \sigma^{-1}(x)+ u \cdot x} =\sum_{y \in \FB^n }(-1)^{f (y)+ u \cdot \sigma (y)},$$
 and clearly $W_{f \circ \sigma^{-1}}(u) \in \{0,\pm 2^{\frac{n+s}{2}} \}$  if and only if $span(\sigma_1 , \ldots, \sigma_n) \subset lp( f )$.
\qed
\end{proof}
Assuming that $f$ given by (\ref{eq:form}) is a bent function, a nonlinear permutation $\sigma$ on $\FB^n$ considered in \cite{HouLang,Houextension2017} which acts nonlinearly on the variables $x_1$ and $x_2$ was given explicitly as:
\begin{equation} \label{eq:sigma}
\sigma(x_1,x_2,\ldots, x_n)= \big ( (f_1,f_2)+ (x_1,x_2) \left [ \begin{array}{cc} 1 & \alpha + 1 \\ \alpha & 1 \end{array} \right ], x_3, \ldots, x_n \big ).
\end{equation}
Then defining
\begin{equation} \label{eq:siginv}
\tau(x_1,x_2,\ldots, x_n)= \big ( [(f_1,f_2)+ (x_1,x_2)] \left [ \begin{array}{cc} 1 & \alpha + 1 \\ \alpha & 1 \end{array} \right ], x_3, \ldots, x_n \big ),
\end{equation}
one can easily verified that $\sigma \circ \tau= id$ which shows that $\tau=\sigma^{-1}$ and consequently $\sigma$ is a permutation.
Furthermore, assuming that  that $f$ is bent  of the form (\ref{eq:form}) it was shown that $F$ defined  as:
\begin{equation} \label{eq:bentF}
F=(\alpha + 1)f_1f_2 + (x_1+1)f_1 + (x_1+x_2 + \alpha + 1)f_2 + \alpha(x_1+1)x_2 + g,
\end{equation}
is also a bent function. Notice that the bentness of  $f$ is  measured as the distance to linear functions, thus verifying that $d_H(f,l)=2^{n-1} \pm 2^{n/2-1}$ for any $l \in \mathcal{L}_n$. On the other hand, the bentness of $F$ is established by measuring the distance  to linear span of $\sigma$. Thus, to show that $F$ is bent is equivalent to showing  that $d_H(f,l)=2^{n-1} \pm 2^{n/2-1}$, for any $l \in span(\sigma_1,\ldots,\sigma_n)$. Furthermore, we have that $span(\sigma_1,\ldots,\sigma_n) \subset span(\mathcal{A}_n, f_1 + x_2\alpha, f_2 + x_1 \alpha)$.

With the following result, we show that the above approach can be applied to plateaued functions as well.
\begin{theo}\label{theo:Ffsigma}
Let $f \in \mathcal{B}_n$ be an $s$-plateaued function (where $n$ and $s$ are of same parity) be given by (\ref{eq:form}). Then, for the permutation $\sigma$ given by (\ref{eq:sigma}) the function $F=f \circ \sigma^{-1}$ defined by (\ref{eq:bentF}) is also an $s$-plateaued function. Furthermore, $f$ and $F$ are EA-inequivalent.  
\end{theo}
\begin{proof} To show that $F$ is $s$-plateaued, by Theorem \ref{th:sigma} it is sufficient to show that for any $l \in span(\sigma_1,\ldots,\sigma_n)$ we have  $W_{f+l}(u) \in \{0, \pm 2^{\frac{n+s}{2}} \}$, where $u\in \mathbb{F}^n_2$. Since $span(\sigma_1,\ldots,\sigma_n) \subset span(\mathcal{A}_n, f_1 + x_2\alpha, f_2 + x_1 \alpha)$ we assume that $l \in span(\mathcal{A}_n, f_1 + x_2\alpha, f_2 + x_1 \alpha)$ so that $l$ can be written as $l=\epsilon_1(f_1 + x_2\alpha) + \epsilon_2(f_2 + x_1 \alpha) + \gamma$, for $\epsilon_1,\epsilon_2 \in \FB$ and $\gamma \in \mathcal{A}_n$. Then, similarly as in \cite{HouLang}, we can write
$$l= f(x_1, \ldots, x_n) + f(x_1 +\epsilon_1 , x_2+\epsilon_2 , x_3 ,\ldots, x_n)+ \epsilon_1\epsilon_2\alpha + \gamma,$$
which implies that  $f+l= f(x_1 +\epsilon_1 , x_2+\epsilon_2 , x_3 ,\ldots, x_n)+ \epsilon_1\epsilon_2\alpha + \gamma$. However,
since $\alpha \in \mathcal{A}_n$ then $\epsilon_1\epsilon_2\alpha + \gamma \in \mathcal{A}_n$, and therefore the Walsh spectra of $f+l$ is just a (linearly) permuted version of the spectra of $f$. Thus, $F$ is $s$-plateaued as well.  \\
That $f$ and $F$ are EA-inequivalent follows from the fact that $\sigma^{-1}$ is a nonlinear permutation, acting nonlinearly on $x_1$ and $x_2$ through $f_1$ and $f_2$ while keeping the remaining coordinates fixed. Thus, there is no affine permutation such that $F(x)=f(Ax+b)$ (with possible addition of affine function) that can achieve this.
\qed 
\end{proof}
With the following example we point out that if $f,F:\mathbb{F}^{n}_2\rightarrow \mathbb{F}_2$ satisfy $F(x)=f(H(x))+\ell(x)$, where $\ell\in\mathcal{A}_n$ and $H:\mathbb{F}^{n}_2\rightarrow \mathbb{F}^n_2$ is a non-linear vectorial function (not necessarily a permutation), then it does not mean that $f$ and $F$ are inequivalent.

The following example demonstrates that in general affine transform of the input function can be in general equal to a non-bijective nonlinear action to the input variables.
\begin{ex}\label{ex:inEA2}
Let $f,F:\mathbb{F}^{5}_2\rightarrow \mathbb{F}_2$ be (two trivial plateaued functions) given as $f(x_1,\ldots,x_5)=x_1x_3 + x_2x_4 + x_5$ and $F(x_1,\ldots,x_5)= x_1x_3 + x_2x_4 + x_2x_5 + x_3x_5 + x_4x_5+x_1+x_4.$ In terms of the non-bijective mapping $H:\mathbb{F}^5_2\rightarrow \mathbb{F}^5_2$ defined as (even though strictly speaking $f(H(x))$ is not well-defined)

$$H(x_1,\ldots,x_5)=(x_1,x_2,x_3,x_4, x_2x_5 + x_3x_5 + x_4x_5)$$ the functions $f$ and $H$ are related as
\begin{eqnarray*}
f(H(x))=(x_1x_3 + x_2x_4)+ (x_2x_5 + x_3x_5 + x_4x_5)=F(x_1,\ldots,x_5)+(x_1+x_4).
\end{eqnarray*}
However, for the matrix $A\in GL(5,\mathbb{F}_2)$ given as
$$A=\left(
                                                           \begin{array}{ccccc}
                                                             1 & 0 & 0 & 0 & 0 \\
                                                             0 & 1 & 0 & 0 & 0 \\
                                                             0 & 0 & 1 & 0 & 0 \\
                                                             0 & 0 & 0 & 1 & 0 \\
                                                             1 & 1 & 0 & 1 & 1 \\
                                                           \end{array}
                                                         \right),$$
one can verify that for $H'(x)=(x_1,\ldots,x_5)A=(x_1+x_5,x_2+x_5,x_3,x_4+x_5,x_5)$ it holds that
\begin{eqnarray*}
f(H'(x))=f(xA)&=&(x_1+x_5)x_3+(x_2+x_5)(x_4+x_5)+x_5\\
&=&(x_1x_3+x_3x_5)+(x_2x_4+x_2x_5+x_4x_5+x_5)+x_5\\
&=&x_1 x_3 + x_2 x_4 + x_2 x_5 + x_3 x_5 + x_4 x_5\\
&=&F(x_1,x_2,x_3,x_4,x_5)+(x_1+x_4),
\end{eqnarray*}
which means that $f$ and $F$ are affine equivalent.
\end{ex}
To illustrate the application of Theorem \ref{theo:Ffsigma}  we consider the semi-bent function $f$ (thus 1-plateaued) in Example 4.1.
\begin{ex}\label{ex:F}
Let $f \in \mathcal{B}_5$ be a semi-bent function   given by $f(x_1,\ldots,x_5)=x_1x_2x_5+ x_1x_3+ x_2x_4+ x_5$ which can be written as $x_1 f_1 + x_2f_2 + x_1x_2 \alpha + g$, where $f_1=x_3$, $f_2=x_4$, $\alpha=x_5$ and $g=x_5$. Apparently, the functions $f_1,f_2, \alpha$ and  $g$ do not depend on the variables $x_1,x_2$ and furthermore $\alpha$ is affine which is the main condition for the construction to work. Then,
\begin{eqnarray*}
F(x)&=& (\alpha + 1)f_1f_2 + (x_1+1)f_1 + (x_1+x_2 + \alpha + 1)f_2 + \alpha(x_1+1)x_2 + g \\
&=& (x_5+1)(x_3x_4) + (x_1+1)x_3 + (x_1+x_2+ x_5+1)x_4 + x_5(x_1+1)x_2 + x_5 \\
&=& x_1x_2x_5 + x_3x_4x_5 + x_1x_3 + x_1x_4 + x_2x_4 + x_2x_5 +x_3x_4 + x_4x_5 + x_3 + x_4 + x_5.
\end{eqnarray*}
The function $F$ was checked to be a plateaued (semi-bent) function  with:
$$W_F=\{0, 8, 0, 8, 0, 8, 8, 0, 0, -8, 0, 8, 0, 8, -8, 0, 0, -8, 0, -8, 0, 8, 8, 0, -8, 0, 8, 0, -8, 0, 0, 8\}.$$
The fact that $f$ and $F$ are EA-inequivalent follows from Theorem \ref{theo:Ffsigma}, but can  also be proved using the results of Section \ref{sec:equivalence} and showing the inequivalence of their supports.
\end{ex}
\begin{rem} One may verify that the Walsh support of the function $F$ (from Example \ref{ex:F}) is $S_F=\mathbb{F}^{4}_2\wr T_{\mu}$, where $\mu(x)=x_1 x_2 + x_3 x_4+1$, and its dual is  given by
$F^*(x)=x_1 + x_2 + x_1 x_2 + x_1 x_3 + x_2 x_3 + x_2 x_4$,  $x\in \mathbb{F}^4_2.$
 It is not difficult to check that there exist many disjoint spectra functions to $F$ (by computer search we find $384$ such functions), and their Walsh supports are given exactly as $\mathbb{F}^4_2\wr T_{\mu+1}$, where the duals are bent functions on $\FB^4$ at bent distance to $\mu$.
\end{rem}
The question whether this nonlinear action can be generalized to provide a set of disjoint
spectra plateaued functions of maximal cardinality seems to be difficult and is left as an open
problem.

\section{Conclusions}\label{sec:OP}

The purpose of this article is to provide further clarity concerning the essential structural
properties of a class of Boolean functions known as plateaued functions. We have provided
an efficient approach of designing these functions in the spectral domain, both trivial and nontrivial ones. The EA-equivalence within and between these classes have been addressed as well as the design of disjoint spectra (regardless of being trivial or not) plateaued functions of maximal cardinality.
The most challenging problem
that remains open is whether the constructed disjoint spectra plateaued functions of maximal
cardinality (whose concatenation is a bent function) may give rise to some new classes of bent
functions.
\\\\
\noindent
{\large \bf Acknowledgment:} Yongzhuang Wei (corresponding author) is supported in part by the
National Key R\&D Program of China (No. 2017YFB0802000)
and in part by the Natural Science Foundation of China (Nos.
61572148, 61872103).  Samir Hod\v zi\' c  is supported in part by the Slovenian Research Agency (research program P3-0384 and Young Researchers Grant). Enes Pasalic is partly supported  by the Slovenian Research Agency (research program P3-0384 and research project J1-9108). For the first two authors, the work is supported in part by H2020 Teaming InnoRenew CoE (grant no. 739574). Fengrong Zhang is supported by Jiangsu Natural Science Foundation (BK20181352).


\section{Appendix}
In the context of Remark \ref{rem:ordering}, we provide an alternative  proof of Theorem \ref{th:EA1}. For this purpose, we also need  to modify Lemma \ref{prop:Hrow1} slightly.

Let $(i)_2$ ($i\in[0,2^{m}-1]$) denote the binary representation of length $m$ of the integer $i$. The linear subspaces $E,E'\subset \mathbb{F}^n_2$ ($\dim(E)=\dim(E')=n-s$) and ${\Bbb F}_2^{n-s}$ we order  as
\begin{eqnarray}\label{eq:EE'}\nonumber
E&=&\{0_n,(1)_2\cdot \mathbf{B}_E,\ldots, (2^{n-s}-1)_2\cdot \mathbf{B}_E\}=\{(i)_2\cdot \mathbf{B}_E: i\in[0,2^{n-s}-1]\},\\
{\Bbb F}_2^{n-s}&=&\{(i)_2\cdot (\beta^{(1)},\beta^{(2)},\ldots,\beta^{(n-s)}): i\in[0,2^{n-s}-1]\},\\\nonumber
E'&=&\{0_n,(1)_2\cdot \mathbf{B}_{E'},\ldots, (2^{n-s}-1)_2\cdot \mathbf{B}_{E'}\}=\{(i)_2\cdot \mathbf{B}_{E'}: i\in[0,2^{n-s}-1]\},
\end{eqnarray}
where the corresponding bases are given as $\mathbf{B}_E=(\alpha^{(1)},\alpha^{(2)},\ldots,\alpha^{(n-s)})$,
$\mathbf{B}_{{\Bbb F}_2^{n-s}}=(\beta^{(1)},\beta^{(2)},$ $\ldots,\beta^{(n-s)})$ and
  $\mathbf{B}_{E'}=(\alpha'^{(1)},\alpha'^{(2)},\ldots,\alpha'^{(n-s)})$.
To clarify the notation, let $n-s=3$. Then the previous notation means that for $(3)_2=(0,1,1)$ and basis $\mathbf{B}_E=(\alpha^{(1)},\alpha^{(2)},\alpha^{(3)})$ we have that $(3)_2\cdot \mathbf{B}_E=\alpha^{(2)}+\alpha^{(3)}.$ Notice that all sets in (\ref{eq:EE'}) satisfy the recursion mentioned in Remark \ref{rem:ordering}.
\begin{lemma}\label{prop:Hrow2}
Let a linear subspace $E\subset {\Bbb F}_2^n$ and space ${\Bbb F}_2^{n-s}$ be ordered as in (\ref{eq:EE'}). Let the mapping $ \sigma:{\Bbb F}_2^{n-s} \rightarrow E$ be defined in terms of indices as
$$\sigma((i)_2\cdot (\beta^{(1)},\beta^{(2)},\ldots,\beta^{(n-s)}))=(i)_2\cdot (\alpha^{(1)},\alpha^{(2)},\ldots,\alpha^{(n-s)}),$$
for fixed bases $\mathbf{B}_E$, $\mathbf{B}_{\mathbb{F}^{n-s}_2}$. Then for an arbitrary matrix $B\in GL(n-s,\mathbb{F}_2)$ and vector $t=c\cdot \mathbf{B}_{\mathbb{F}^{n-s}_2} \in \mathbb{F}^n_2, c\in {\Bbb F}_2^{n-s}$, there exists a matrix $D\in GL(n,\mathbb{F}_2)$ and vector $\gamma\in E$ such that
 $$\sigma\left(\left((i)_2\cdot (\beta^{(1)},\beta^{(2)},\ldots,\beta^{(n-s)})\right)B+t\right)=\left((i)_2\cdot (\alpha^{(1)},\alpha^{(2)},\ldots,\alpha^{(n-s)})\right)D+\gamma.$$
 Furthermore, $\left((i)_2\cdot (\alpha^{(1)},\alpha^{(2)},\ldots,\alpha^{(n-s)})\right)D+\gamma \in E$ for $i\in[0,2^{n-s}-1]$.
\end{lemma}
\begin{proof}
If there exist integers $i_{1},i_{2},\ldots, i_{{n-s}}$ such that the vectors
$$(i_{1})_2\cdot \mathbf{B}_{{\Bbb F}_2^{n-s}}, (i_{2})_2\cdot \mathbf{B}_{{\Bbb F}_2^{n-s}},\ldots, (i_{{n-s}})_2\cdot \mathbf{B}_{{\Bbb F}_2^{n-s}}$$
constitute a basis of ${\Bbb F}_2^{n-s}$, then $\left((i_{1})_2,(i_{2})_2, \ldots, (i_{{n-s}})_2 \right)$ is a basis of ${\Bbb F}_2^{n-s}$. Furthermore, the vectors
$$(i_{1})_2\cdot \mathbf{B}_E, (i_{2})_2\cdot  \mathbf{B}_E,\ldots, (i_{{n-s}})_2\cdot  \mathbf{B}_E$$ also constitute a basis of $E$, since the pair of basis $\mathbf{B}_{{\Bbb F}_2^{n-s}}$ and $\mathbf{B}_E$ are related via the same linear combinations of basis vectors determined by $(i)_2$, where $i\in[0,2^{n-s}-1]$.

Without loss of generality, let the vectors $\beta^{(j)}$ and $\beta^{(i_{j})}$  be related as
 $$ \beta^{(j)}B+t= \beta^{(i_{j})}\in \mathbb{F}^{n-s}_2,\;\;\;j\in[1,n-s],$$
where $B\in GL(n-s,\mathbb{F}_2)$ and $t=c\cdot \mathbf{B}_{{\Bbb F}_2^{n-s}}$ (for some vector $c\in {\Bbb F}_2^{n-s}$). Clearly, since  $(\beta^{(i_1)},\beta^{(i_2)},\ldots,\beta^{(i_{n-s)}})$ is a basis of ${\Bbb F}_2^{n-s}$, then $(\alpha^{(i_1)},\alpha^{(i_2)},\ldots,\alpha^{(i_{n-s)}})$ is a basis of $E$.

Denoting by $\widetilde{\mathbf{B}}_E=(\alpha^{(i_1)},\alpha^{(i_2)},\ldots,\alpha^{(i_{n-s)}})$, there exists a matrix $D\in GL(n,\mathbb{F}_2)$ such that
$$\widetilde{\mathbf{B}}_E=\mathbf{B}_ED+c\cdot\mathbf{B}_E.$$
Considering the basis ${\mathbf{B}}_{{\Bbb F}_2^n}=(\alpha^{(1)},\ldots,\alpha^{({n-s)}}, \alpha^{({n-s+1)}},\ldots,\alpha^{({n)}})$, and since we view $\alpha^{(i_1)},\ldots,\alpha^{(i_{n-s})}$ as linear combinations of vectors $\alpha^{(1)},\ldots,\alpha^{(n-s)}$, then $\widetilde{\mathbf{B}}_{{\Bbb F}_2^n}=(\alpha^{(i_1)},\ldots,\alpha^{(i_{n-s)}}, \alpha^{({n-s+1)}},\ldots,\alpha^{({n)}})$ is a basis of ${\Bbb F}_2^n$ as well.
Denoting by
$ D=\left(
    \begin{array}{cccc}
      d_{11} &   d_{12} & \ldots &   d_{1n} \\
      d_{21} &   d_{22} & \ldots &   d_{2n} \\
     \vdots &   \vdots & \ddots &   \vdots \\
     d_{n1} &  d_{n2} & \ldots &   d_{nn} \\
    \end{array}
  \right),
$
and $\alpha^{(j)}=(\alpha^{(j)}_1,\alpha^{(j)}_2,\ldots,\alpha^{(j)}_n)$, the system
$$\left(
    \begin{array}{c}
      \alpha^{(i_1)} \\
      \vdots \\
      \alpha^{(i_{n-s)}}\\
       \alpha^{({n-s+1)}}\\
       \vdots\\
      \alpha^{({n)}}  \\
    \end{array}
  \right)=\left(
    \begin{array}{c}
      \alpha^{(1)} \\
      \vdots \\
      \alpha^{({n-s)}}\\
       \alpha^{({n-s+1)}}\\
       \vdots\\
      \alpha^{({n)}}  \\
    \end{array}
  \right)D+c'\cdot \mathbf{B}_{{\Bbb F}_2^n}
 $$
 means that
 $$
 \left\{
   \begin{array}{l}
     \alpha^{(1)}_1d_{1k}+\alpha^{(1)}_2d_{2k}+\ldots+\alpha^{(1)}_nd_{nk}+\gamma_1=\alpha^{(i_1)}_k, \\
    \alpha^{(2)}_1d_{1k}+\alpha^{(2)}_2d_{2k}+\ldots+\alpha^{(2)}_nd_{nk}+\gamma_2=\alpha^{(i_2)}_k, \;\;\;\;\;k\in[1,n],\\
\vdots\\
 \alpha^{(n)}_1d_{1k}+\alpha^{(n)}_2d_{2k}+\ldots+\alpha^{(n)}_nd_{nk}+\gamma_n=\alpha^{(i_n)}_k,
   \end{array}
 \right.
 $$
where for $c'=(c,0,\ldots,0)\in {\Bbb F}_2^n$ the vector $\gamma$ is given as $\gamma=c'\cdot{\mathbf{B}}_{{\Bbb F}_2^n}=c\cdot\mathbf{B}_E\in E$. This means that the matrix $D$ can be computed.
%
Furthermore, denoting by $\widetilde{\mathbf{B}}_{\mathbb{F}^{n-s}_2}=(\beta^{(i_1)},\beta^{(i_2)},\ldots,\beta^{(i_{n-s})})$ we  have that
$$\sigma((i)_2\cdot \widetilde{\mathbf{B}}_{\mathbb{F}^{n-s}_2})=\sigma
((i)_2\cdot \mathbf{B}_{\mathbb{F}^{n-s}_2}B+t)=(i)_2\cdot\mathbf{B}_ED+\gamma,$$
holds for all $i\in[0,2^{n-s}-1]$. Since $\mathbf{B}_ED$ is also a basis of $E$ and $\gamma \in E$, we have $(i)_2\cdot\mathbf{B}_ED+\gamma\in E$
 which completes the proof.\qed
\end{proof}
\begin{theo}
Two trivial $s$-plateaued functions $f,h:\mathbb{F}^n_2\rightarrow \mathbb{F}_2$ are EA-equivalent if and only if their duals $\overline{f}^*,\overline{h}^*:\mathbb{F}^{n-s}_2\rightarrow\mathbb{F}_2$ 
are EA-equivalent bent functions. 
\end{theo}
\begin{proof} $(\Rightarrow)$ By \cite[Theorem 3.1-$(ii)$]{Secondary} and relations (\ref{eq:identif})-(\ref{eq:equiv22}) it is clear that $\overline{f}^*,\overline{h}^*:\mathbb{F}^{n-s}_2\rightarrow\mathbb{F}_2$  are EA-equivalent bent functions.

$(\Leftarrow)$ In this part, we assume that $f$ and $h$ are two fixed trivial plateaued functions and thus the values of $h^*$ (on $S_h$) and $f^*$ (on $S_f$) are fixed in advance. In this context, we need to specify their values on $\mathbb{F}^{n-s}_2$.

 Suppose that $S_f$ and $S_h$ are given as $S_f=v+E$ and $S_h=v'+E'$ respectively,   where  $E$ and $E'$ are defined as (\ref{eq:EE'}). Without loss of generality, denote  $E= \{e_0,\ldots,e_{2^{n-s}-1}\}$ and $E= \{e'_0,\ldots,e'_{2^{n-s}-1}\}$, that is $e_i= (i)_2\cdot \mathbf{B}_E $  and $e'_i= (i)_2\cdot \mathbf{B}_{E'} $.
 With respect to $S_f=v+E$, the dual $\overline{f}^*:\mathbb{F}^{n-s}_2\rightarrow \mathbb{F}_2$ is defined as
\begin{equation}\label{equ add1}\overline{f}^*(x_i)=f^*(v+e_i)=f^*(\omega_i),\;\;\;x_i\in \mathbb{F}^{n-s}_2,\;e_i\in E.\end{equation}
 With respect to $S_h=v'+E'$, the dual $\overline{h}^*:\mathbb{F}^{n-s}_2\rightarrow \mathbb{F}_2$ is defined as
\begin{equation}\label{equ add2}\overline{h}^*(x_i)=h^*(v'+e'_i)=h^*(z_i),\;\;\;x_i\in \mathbb{F}^{n-s}_2,\;e'_i\in E'.\end{equation}
Due to Proposition \ref{prop:Hrow} and Theorem \ref{theo:plateH}-$(ii)$, it is not difficult to see that both duals $\overline{h}^*$ and $\overline{f}^*$ are bent functions.

Now, let us assume that $\overline{f}^*,\overline{h}^*$  are affine equivalent. That is, there exist a matrix $B\in GL(n-s,\mathbb{F}_2)$ ant two vectors $t,r\in \mathbb{F}^{n-s}_2$ and $\kappa\in \mathbb{F}_2$ such that  $$\overline{h}^*(x_i)=\overline{f}^*(x_iB+t)+r\cdot x_i+\kappa, \;\;\;\;\; x_i\in \mathbb{F}^{n-s}_2.$$
By Proposition \ref{prop:Hrow}, there  exists  $c\in \mathbb{F}^n_2$ such that $c\cdot \omega_i=r\cdot x_i$ holds for all $i\in[0,2^{n-s}-1]$, i.e.,
$$(c\cdot \omega_0,\ldots,c\cdot \omega_{2^{n-s}-1})=(r\cdot x_0,\ldots,r\cdot x_{2^{n-s}-1})$$
is truth table of a linear function in $n-s$ variables ($S_f$ is an affine subspace). Furthermore, by Lemma \ref{prop:Hrow2}, (\ref{equ add1}) and  (\ref{equ add2}),  we have
  \begin{eqnarray}\label{equ add3}
  {h}^*(z_i)&=& \overline{h}^*(x_i)=\overline{f}^*(x_iB+t)+r\cdot x_i+\kappa=f^*(v+(e_iD+\gamma ))+c\cdot \omega_i+\kappa,
  \end{eqnarray}
where $D$ and $ \gamma$ can be obtained by $B$ and $t$. Note that by Lemma \ref{prop:Hrow2} we have that $e_iD+\gamma\in E$ and  consequently $v+(e_iD+\gamma)\in S_f$.

Now setting ${\mathbf{B}}_{{\Bbb F}_2^n}=(\alpha^{(1)},\ldots,\alpha^{({n-s)}}, \alpha^{({n-s+1)}},\ldots,\alpha^{({n)}})$ and $\widetilde{\mathbf{B}}_{{\Bbb F}_2^n}=(\alpha'^{(1)},\ldots,\alpha'^{({n-s)}},$ $ \alpha'^{({n-s+1)}},\ldots,\alpha'^{({n)}})$, there exists a matrix $A\in GL(n,\mathbb{F}_2)$
such that
$$ \left(
     \begin{array}{c}
       \alpha^{(1)}\\
       \vdots \\
       \alpha^{({n)}} \\
     \end{array}
   \right)=  \left(
     \begin{array}{c}
       \alpha'^{(1)}\\
       \vdots \\
       \alpha'^{({n)}} \\
     \end{array}
   \right)A,
$$
holds, i.e., we have $\alpha^{(i)}=\alpha'^{(i)}A$, $i=1,2,\ldots, n$.  Furthermore, we have $e_i=e'_i A$ for $i=0,1,2\ldots,2^{n-s}-1$. Now, for $x\in \mathbb{F}^n_2$, by (\ref{WHT}), (\ref{equ add1}), (\ref{equ add2}) and  (\ref{equ add3}) we have
 \begin{equation}\label{equ IWHT}
 \begin{array}{rl}
 (-1)^{h(x)}=&2^{-n}\sum\limits_{e'_i\in E'}W_h(e'_i+v')(-1)^{e'_i\cdot x+v'\cdot x}\\
=&2^{\frac{s-n}{2}}\sum\limits_{e'_i\in E'}(-1)^{h^*(e'_i+v')+e'_i\cdot x+v'\cdot x}\\
=&2^{\frac{s-n}{2}}\sum\limits_{e_i\in E}(-1)^{f^*(v+e_iD+\gamma)+e_iA^{-1}\cdot x+v'\cdot x}\\
=&2^{\frac{s-n}{2}}\sum\limits_{e_i\in E}(-1)^{f^*(\omega_iD+vD+v+\gamma )+c\cdot \omega_i+\kappa+e_iA^{-1}\cdot x+v'\cdot x}\\
=&2^{-n}\sum\limits_{\omega_i\in S_f}W_f(\omega_iD+vD+v+\gamma)(-1)^{e_iA^{-1}\cdot x+v'\cdot x+c\cdot \omega_i+\kappa}\\
\stackrel{t_i=\omega_iD+vD+v+\gamma}{=} &2^{-n}\sum\limits_{t_i\in S_f}W_f(t_i)(-1)^{[(t_i+\gamma+v)D^{-1}A^{-1}+v']\cdot x+c\cdot (t_i+vD+\gamma+v)D^{-1}+\kappa}\\
{=} &2^{-n}\sum\limits_{t_i\in S_f}W_f(t_i)(-1)^{t_i\cdot [x(D^{-1}A^{-1})^T+cD^{-T}]+[(\gamma+v)D^{-1}A^{-1}+v']\cdot x+[c\cdot (vD+\gamma+v)D^{-1}+\kappa]}\\
{=}& (-1)^{f(x(D^{-1}A^{-1})^T+cD^{-T})+[(\gamma+v)D^{-1}A^{-1}+v']\cdot x+[c\cdot (vD+\gamma+v)D^{-1}+\kappa]}.\\
 \end{array}
 \end{equation}
Hence $f$ and $h$ are EA-equivalent, and the proof is completed.\qed
\end{proof}

\end{document}